\documentclass[reqno]{amsart}
\usepackage{biblatex}
\usepackage{amsmath,amsfonts,amssymb}
\usepackage{hyperref}
\usepackage{orcidlink}
\usepackage{fullpage}
\usepackage{algorithm}
\usepackage[noend]{algpseudocode}

\newtheorem{lemma}{Lemma}

\theoremstyle{definition}
\newtheorem{definition}{Definition}
\newtheorem{protocol}{Protocol}
\newtheorem{example}{Example}
\newtheorem{remark}{Remark}

\title{Marginal sets in semigroups and semirings}

\author{I. Buchinskiy}
\address{I. Buchinskiy\,\orcidlink{0000-0001-8637-9127}, Sobolev Institute of Mathematics, Omsk, Russia}
\email{buchvan@mail.ru}

\author{M. Kotov}
\address{M. Kotov\,\orcidlink{0000-0002-2820-1053}, Independent Researcher\footnote{The second author conducted the research while working at the Sobolev Institute of Mathematics.}}
\email{matvej.kotov@gmail.com}

\author{A. Ponmaheshkumar}
\address{A. Ponmaheshkumar\,\orcidlink{0009-0005-4744-5247},
Department of Mathematics, SRM Institute of Science and Technology, Kattankulathur, Tamil Nadu, India}
\email{pa3258@srmist.edu.in}

\author{R. Perumal}
\address{R. Perumal\,\orcidlink{0000-0002-7137-0485},
Department of Mathematics, SRM Institute of Science and Technology, Kattankulathur, Tamil Nadu, India}
\email{perumalr@srmist.edu.in}

\begin{document}
\begin{abstract}
In 2019, V.~A.~Roman'kov introduced the concept of marginal sets for groups. He developed a theory of marginal sets and demonstrated how these sets can be applied to improve some key exchange schemes. In this paper, we extend his ideas and introduce the concept of marginal sets for semigroups and semirings. For tropical matrix semigroups and semirings, we describe how some marginal sets can be constructed. We apply marginal sets to improve some key exchange schemes over semigroups.
\end{abstract}

\date{\today}

\maketitle

\section{Introduction}

The concept of marginal sets for groups, introduced by Roman'kov, was inspired by the classical notion of marginal subgroups. However, the generalization to marginal sets is very different from the original concept. For context, recall the classical definition of a marginal subgroup introduced by P.~Hall~\cite{Hall1940}.

\begin{definition}
Let $\mathcal{G}$ be a group and let $w(x_1, \ldots, x_n)$ be a group word. A normal subgroup $\mathcal{N}$ is said to be $w$-\textit{marginal} in $\mathcal{G}$ if 
$$
w(ag_1,  g_2, \ldots, g_n) = 
w(g_1,  ag_2, \ldots, g_n) = 
\ldots
w(g_1,  g_2, \ldots, ag_n) = 
w(g_1,  g_2, \ldots, g_n)
$$
for all $a \in \mathcal{N}$, $g_i \in G$, $1 \leq i \leq n$.
\end{definition}

V. Roman'kov introduced the following generalization of this definition~\cite{Romankov2019A, Romankov2019B, Romankov2019C, Romankov2021, Romankov2022}.

\begin{definition}
Let $\mathcal{G}$ be a group, $w(x_1, \ldots, x_n)$ be a group word, and $g = (g_1, \ldots, g_n)$ be a tuple of elements of $\mathcal{G}$. A tuple $c = (c_1, \ldots, c_n)$ of elements of $\mathcal{G}$ is called a \textit{marginal tuple} with respect to $w$ and $g$ if 
$$
w(c_1g_1, \ldots, c_ng_n) = w(g_1, \ldots, g_n).
$$
We denote this relation by $c \perp w(g)$. A set $C$ of $n$-tuples of elements of $\mathcal{G}$ is said to be \textit{marginal} with respect to $w$ and $g$ if $c \perp w(g)$ for every tuple $c \in C$. In this case, we write $C \perp w(g)$.
\end{definition}

\begin{example}
For example, let $w = [x_1, x_2]$ and let $g = (g_1, g_2)$ be an arbitrary pair of elements of a group $\mathcal{G}$. Any tuple of the form $c = (c_1, 1),  c_1 \in G$ such that $[c_1, g_2] = 1$ is marginal with respect to $w$ and $g$, because 
\begin{multline*}
[c_1g_1, g_2] = 
(c_1g_1)^{-1} g_2^{-1} c_1g_1 g_2 =
g_1^{-1} c_1^{-1} g_2^{-1} c_1 g_1 g_2 = 
g_1^{-1} c_1^{-1} g_2^{-1} c_1 g_2 g_2^{-1} g_1 g_2 = 
g_1^{-1} [c_1, g_2] g_2^{-1} g_1 g_2 = \\
g_1^{-1} g_2^{-1} g_1 g_2 = 
[g_1, g_2].
\end{multline*}
\end{example}

Note that a marginal set does not have to be a subgroup. It can be chosen to be a very wild set. For example, it does not even have to be recursive computable. There are different ways to construct marginal sets. 
Since any solution $(c_1, \ldots c_n)$ to the equation 
$
w(x_1g_1, \ldots, x_ng_n) = w(g_1, \ldots, g_n)
$
is a marginal tuple with respect to $w$ and $g$, 
Roman'kov considered two randomized algorithms for constructing marginal sets: Algorithm~\ref{Algo1} and Algorithm~\ref{Algo2}.

\begin{algorithm}
\caption{Generate a set $C$ such that $C \perp w(g)$} \label{Algo1}
\begin{algorithmic}
\State \textbf{Input:} A group $\mathcal{G}$, a word $w(x_1, \ldots, x_n) = x_1 \cdots x_n$, a tuple $g = (g_1, \ldots, g_n)$ of elements of $\mathcal{G}$, and a~number of iterations $m$
\State \textbf{Output:} A set $C$ such that $C \perp w(g)$
\State Initialize $C \gets \emptyset$
\For{$k = 1$ to $m$}
    \State Randomly select $i \in \{1, \ldots, n\}$
    \State Randomly select values $c_j \in \mathcal{G}$ for $j \neq i$ 
    \State $c_i \gets g_{i - 1}^{-1}c_{i - 1}^{-1}\cdots g_1^{-1}c_1^{-1}g_1 \cdots g_n g_n^{-1}c_n^{-1}\cdots g_{i+1}^{-1}c_{i + 1}^{-1}$
    \State $C \gets C\cup \{(c_1, \ldots, c_n)\}$
\EndFor
\State \Return $C$
\end{algorithmic}
\end{algorithm}

\begin{algorithm}
\caption{Generate a tuple $c'$ such that $c' \perp w(g)$}\label{Algo2}
\begin{algorithmic}
\State \textbf{Input:} A group $\mathcal{G}$, a word $w(x_1, \ldots, x_n) = x_1 \cdots x_n$, a tuple $g = (g_1, \ldots, g_n)$ of elements of $\mathcal{G}$, a marginal tuple $c = (c_1, \ldots, c_n)$ with respect to $w$ and $g$, and a number of iterations $m$.
\State \textbf{Output:} A tuple $c'$ such that $c' \perp w(g)$.

\State Initialize $c' \gets c$
\For{$k = 1$ to $m$}
    \State Randomly select $i \in \{1, \ldots, n\}$
    \State Randomly select $d \in \mathcal{G}$
    \State $c'_i \gets c'_i g_i d g_i^{-1}$
    \State $c'_{i+1} \gets d^{-1} c'_{i+1}$
\EndFor
\State \Return $c'$
\end{algorithmic}
\end{algorithm}

In~\cite{Romankov2019A}, Roman'kov employed marginal sets to improve the Anshel--Anshel--Goldfeld protocol. Here, we present a simplified version of his protocol, focusing only on the use of marginal sets and omitting additional improvements such as permutations of published elements, virtual elements, and hidden elements.

\begin{protocol}
Suppose that two correspondents, Alice and Bob, want to exchange a key. They agree on a public group $\mathcal{G} = \left<G, \cdot, {}^{-1}, e\right>$ explicitly given by a finite set of generators. We suppose that the main group operations can be computed efficiently and that the word problem for $\mathcal{G}$ is solvable efficiently.
\begin{enumerate}
\item 
Alice fixes a positive integer $k$ and chooses a public tuple of elements $a = (a_1, \ldots, a_k)$. She selects a private group word $u = u(x_1, \ldots, x_k)$ and computes $u(a) = u(a_1, \ldots, a_n)$. She constructs a public marginal set $C \subseteq G^k$, $C \perp u(a)$. She publishes the tuple $a$ and the set $C$.
\item
Bob fixes a positive integer $l$ and chooses a public tuple of elements $b = (b_1, \ldots, b_l)$. He selects a private group word $v = v(y_1, \ldots, v_l)$ and computes $v(b) = v(b_1, \ldots, b_l)$. He constructs a public marginal set $D \subseteq G^l$, $D \perp v(b)$. He publishes the tuple $b$ and the set $D$.

\item Alice selects a tuple $d \in D$ and computes $db = (d_1b_1, \ldots, d_lb_l)$. She sends Bob the tuple $$(db)^{u(a)} = ((d_1b_1)^{u(a)}, \ldots, (d_lb_l)^{u(a)}).$$

\item Bob selects a tuple $c \in C$ and computes $ca = (c_1a_1, \ldots, c_ka_k)$. He sends Alice the tuple 
$$(ca)^{v(b)} = ((c_1a_1)^{v(b)}, \ldots, (c_ka_k)^{v(b)}).$$

\item Alice computes $K = u(a)^{-1}u((ca)^{v(b)}) $.

\item Bob computes $K = (v((db)^{u(a)}))^{-1}v(b)$.
\end{enumerate}

Note that they
share the same key, since $ u(a)^{-1}u((ca)^{v(b)}) = (v((db)^{u(a)}))^{-1}v(b) = [u(a), v(b)]$.
\end{protocol}

Roman'kov extended the concept of marginal sets to associative algebras~\cite{Romankov2022}. Recall that an \textit{algebra word} in variables $x_1, \ldots, x_n$ is a linear combination of monomials in $x_1, \ldots, x_n$.

\begin{definition}
Let $\mathcal{M} = \left<M, {+}, {-}, {\cdot}, 0, 1\right>$ be an associative algebra over a field $\mathcal{F}$, $w = w(x_1, \ldots, x_n)$ be an algebra word, and $g = (g_1, \ldots, g_n) \in M^n$ be a tuple of elements of $M$. We say that $c = (c_1, \ldots, c_n) \in M^n$ is a \textit{marginal tuple} with respect to $w$ and $g$ if 
$$
w(c_1g_1, \ldots c_ng_n) = w(g_1, \ldots, g_n).
$$
We denote this relation by $c \perp w(g)$ in this case. A set $C \in M^n$ is said to be \textit{marginal} with respect to $w$ and $g$ if
$c \perp w(g)$ for every tuple $c \in C$. In this case, we also write $C \perp w(g)$.
\end{definition}

\begin{example}
Let $w(x_1, x_2) = x_1^2 - x_2$ be an algebra word. For any pair $(g_1, g_2) \in M^2$ with invertible $g_2$, a tuple of the form $c = (c_1, c_1g_1c_1g_1g_2^{-1} - g_1^2g_2^{-1} + 1) \in M^2$ is marginal with respect to $w(x_1, x_2)$ and $(g_1, g_2)$. Indeed,
$w(c_1g_1, c_2g_2) = 
w(c_1g_1, (c_1g_1c_1g_1g_2^{-1} - g_1^2g_2^{-1} + 1)g_2) = 
w(c_1g_1, (c_1g_1)^2 - g_1^2 + g_2) = 
(c_1g_1)^2 - (c_1g_1)^2 + g_1^2 - g_2 =
g_1^2 - g_2 = w(g_1, g_2).
$
\end{example}

Roman'kov offered Algorithm~\ref{Algo3} for constructing marginal sets in associative algebras.

\begin{algorithm}
\caption{Generate a set $C$ such that $C \perp w(g)$} \label{Algo3}
\begin{algorithmic}
\State \textbf{Input:} An associative algebra $\mathcal{M}$, an algebra word $w(x_1, \ldots, x_n)$, a tuple $g = (g_1, \ldots, g_n)$ of elements of $\mathcal{M}$, and a number of iterations $m$
\State \textbf{Output:} A set $C$ such that $C \perp w(g)$
\State Initialize $C \gets \emptyset$
\For{$k = 1$ to $m$}
    \State Randomly select $i \in \{1, \ldots, n\}$ such that $x_i$  can be expressed in the form
$$x_i = x_i(x_1, \ldots, x_{i - 1}, x_{i + 1}, \ldots, x_n, g_1, \ldots, g_n)$$
    \State Randomly select values $c_j \in \mathcal{M}$ for $j \neq i$ 
    \State $c_i \gets  x_i(c_1, \ldots, c_{i - 1}, c_{i + 1}, \ldots, c_n, g_1, \ldots, g_n)$
    \State $C \gets C\cup \{(c_1, \ldots, c_n)\}$
\EndFor
\State \Return $C$
\end{algorithmic}
\end{algorithm}

Roman'kov used marginal sets to improve the Diffie--Hellman protocol in~\cite{Romankov2022}. Here, we also present a simplified version of this protocol,  focusing only on the use of marginal sets.

\begin{protocol}
Let $\mathcal{G}$ be a group. Let $\mathcal{A}, \mathcal{B}$ be finitely generated elementwise permutable subgroups of $\mathcal{G}$. Consider the algebra $\mathcal{M} = \mathrm{Alg}(\mathcal{G})$ generated by $\mathcal{G}$ and fix an element $g \in \mathcal{M}$.
$\mathrm{Fix}(\mathcal{A}) = \{f \in \mathcal{G} : f^a = f \text{ for all } a \in \mathcal{A}\}$,
$\mathrm{Fix}(\mathcal{B}) = \{h \in \mathcal{G} :  h^b = h \text{ for all } b \in \mathcal{B}\}$. These structures are public.

\begin{enumerate}
\item Alice chooses an algebra word $u = u(x_1, \ldots, x_k)$ and a tuple $a = (a_1, \ldots, a_k)$ of elements of $\mathcal{M}$ such that $g = u(a_1, \ldots, a_k)$. She also chooses a tuple of elements $h = (h_1, \ldots, h_k) \in \mathrm{Fix}(\mathcal{B})^k$ and $ah = (a_1h_1, \ldots, a_kh_k)$. Alice constructs a marginal set $C$, $C \perp u(a_1, \ldots, a_k)$. She publishes $ah$ and $C$.

\item Bob chooses an algebra word $v = v(y_1, \ldots, y_l)$ and a tuple $b = (b_1, \ldots, b_l)$ of elements of $\mathcal{M}$ such that $g = v(b_1, \ldots, b_l)$. He also chooses a tuple of elements $f = (f_1, \ldots, f_l) \in \mathrm{Fix}(\mathcal{A})^l$ and computed $bf$. He
constructs a marginal set $D$, $D \perp v(b_1, \ldots, b_l)$.
He publishes $bf$ and $D$.

\item 
Alice chooses a random element $s \in \mathcal{A}$ and a random tuple $d \in D$.
She computes $$(dbf)^s = ((d_1b_1f_1)^s, \ldots, (d_lb_lf_l)^s)$$ and sends the result to Bob.

\item Bob chooses a random element $t \in \mathcal{B}$ and a random tuple $c \in C$.
He computes $$
(c a h)^t = ((c_1 a_1 h_1)^t, \ldots,(c_k a_k h_k)^t)
$$ and sends the result to Alice.

\item Alice computes $K = u((cah)^t h^{-1})^s$.

\item Bob computes $K = v((dbf)^s f^{-1})^t$.
\end{enumerate}

They share the same key because $u((cah)^t h^{-1})^s = u((ca)^t h^t h^{-1})^s =
u((ca)^t h h^{-1})^s  = u((ca)^t))^s = (u(ca))^{ts} = (u(a))^{ts} = g^{ts}$ and
$v((dbf)^s f^{-1})^t = v((db)^s f^s f^{-1})^t =
v((db)^s f f^{-1})^t = v((db)^s))^t = (v(db))^{st} = (v(b))^{st} = g^{st} = g^{ts}$.
\end{protocol}

Let $\mathcal{G} = \left<G, {\cdot}\right>$ be a semigroup. Recall that a set $H \subseteq G$ 
is called pairwise commuting if $a \cdot b = b \cdot a$ for all $a, b \in H$. 

Sidelnikov, Cherepnev, and Yaschenko proposed the following key exchange method based on non-commutative semigroups~\cite{SidelnikovCherepnevYaschenko1994}. 
Note that a similar idea of using non-abelian groups was proposed by Stickel~\cite{Stickel2005}.

\begin{protocol}\label{ProtocolSidelnikov}
Let $\mathcal{G}$ be a non-commutative semigroup, ${H}$ and
${R}$ be pairwise commuting subssets of $\mathcal{G}$, and $W \in \mathcal{G}$. 
\begin{enumerate}
\item Alice chooses two elements $p_1 \in {H}$ and $q_1 
\in {R}$ as her secret key. She
computes $u = p_1 \cdot w \cdot q_1$ and sends it to Bob.

\item Bob chooses two elements $p_2 \in {H}$ and $q_2 \in {R}$ as his secret key. He
computes $v = p_2 \cdot w \cdot q_2$ and sends it to Alice.

\item Alice computes the common secret key $k = p_1 \cdot v \cdot q_1$.

\item Bob computes the common secret key $k = p_2 \cdot u \cdot q_2$.
\end{enumerate}
They share the same key because $p_1 \cdot (p_2 \cdot w \cdot q_2) \cdot q_1 = p_2 \cdot (p_1 \cdot w \cdot q_1) \cdot q_2$.
\end{protocol}

The success of this method is determined by the choice of $\mathcal{G}$, ${H}$, and ${R}$. Grigoriev and Shpilrain~\cite{GrigorievShpilrain2014} introduced the idea of using matrices over tropical semirings in this scheme. Two classical tropical semirings are the max-plus algebra $\mathcal{R}_{{\max}, {+}} = \left< \mathbb{R} \cup \{-\infty\}, \oplus, \otimes \right>$, where \( a \otimes b = a + b \) and \( a \oplus b = \max(a, b) \), and the min-plus algebra $\mathcal{R}_{{\min}, {+}} = \left<\mathbb{R} \cup \{\infty\}, \oplus, \otimes\right>$, where \( a \otimes b = a + b \) and \( a \oplus b = \min(a, b) \). There are two reasons for this choice. First, it helps avoid linear algebra attacks. Second, the operations can be performed quickly and efficiently. 

This approach is now known as tropical cryptography, and numerous papers have been published in this area in recent years. Below, we highlight some of these papers.

In the original paper by Grigoriev and Shpilrain~\cite{GrigorievShpilrain2014}, tropical matrix polynomials were used. Kotov and Ushakov \cite{KotovUshakov2018} demonstrated that this original tropical key exchange scheme is insecure.
Huang et al.~\cite{Huang} and Amutha and Perumal~\cite{Amutha} proposed using tropical upper-$t$-circulant, tropical lower-$t$-circulant, and tropical anti-$t$-$p$-circulant matrices.
Attacks on these protocols were offered by Buchinskiy et al.~\cite{BuchinskiyKotovTreier2024A} and Alhussaini et al.~\cite{AlhussainiCollettSergeev2023}.
Last year, Alhussaini and Sergeev~\cite{AlhussainiSergeev2024,sergev,alhussaini2024security} analyzed various generalizations of this scheme, using mixed-integer linear programming and optimization heuristics.
Recently, Otero S{\'a}nchez~\cite{Otero2024} considered a new approach for analyzing similar schemes.

Grigoriev and Shpilrain in~\cite{Grigoriev2} offered a new scheme based on extensions of tropical matrix algebras by homomorphisms in an attempt to destroy patterns in powers of elements of a platform algebra. This approach has been analyzed by Rudy and Monico \cite{Rudy}, Isaac and Kahrobaei~\cite{Issac}, Muanalifah and Sergeev~\cite{Muanalifah}, and Jackson and Perumal~\cite{jack}. 

Besides the classical tropical structures, max-times, min-times, digital, and other semirings have been proposed \cite{Huang1, mahesh1, ponmaheshkumar2023robust, mahesh2, Durcheva2014}. The analyses of these schemes can be found in~\cite{AlhussainiCollettSergeev2024,PonmaheshkumarKotovPerumal2025, BuchinskiyKotovTreier2025B}. Ponmaheshkumar and Ramalingam~\cite{ponmaheshkumar2024multi} and Ponmaheshkumar and Perumal~\cite{mahesh2} considered schemes based on supertropical semirings.

Recently, Chen et al.~\cite{ChenGrigorievShpilrain2023} suggested using tropical algebras as platforms for a digital signature protocol.

The analysis of tropical protocols has inspired new research problems in tropical mathematics. Buchinskiy et al. investigated the complexity of tropical knapsack-type problems~\cite{BuchinskiyKotovTreier2025A} and the complexity of solving tropical polynomial systems of degree two~\cite{BuchinskiyKotovTreier2024B}.

For more information on tropical algebra, we refer the reader to~\cite{Butkovic2010}. For more information on
non-commutative cryptography, in which different non-commutative algebraic structures such as non-commutative groups, semigroups, and rings are used as platforms for cryptographic protocols, we refer
the reader to~\cite{MyasnikovShpilrainUshakov2011}.

Let us go back to Protocol~\ref{ProtocolSidelnikov}.   
One way to attack this scheme is as follows.

Suppose that $p_1$ and $q_1$ can be presented as linear combinations:
$$p_1 = \sum_{i=1}^m{x_i}a_i \quad and \quad q_1 = \sum_{j=1}^n{y_j}b_j.$$ 
Then,
$$u = \left(\sum_{i=1}^m{x_i}a_i\right) \cdot w \cdot \left(\sum_{j=1}^n{y_j}b_j\right) = 
\sum_{i=1}^m\sum_{j=1}^n {x_i} \cdot y_j \cdot (a_i \cdot w \cdot b_j).$$
Of course, we can try to find $x_i$ and $y_j$ solving corresponding system of equations. However, it is possible to do the following trick. Suppose that we can find $z_{ij}$ such that
$$u = \sum_{i=1}^m\sum_{j=1}^n z_{ij} \cdot (a_i \cdot w \cdot b_j).$$
Then, we can compute the secret key in the following way:
$$\sum_{i=1}^m\sum_{j=1}^n z_{ij} \cdot a_i \cdot v \cdot b_j = K.$$
Indeed, this works because
$$\sum_{i=1}^m\sum_{j=1}^n z_{ij} \cdot a_i \cdot v \cdot b_j =
\sum_{i=1}^m\sum_{j=1}^n z_{ij} \cdot a_i \cdot p_2 \cdot w \cdot q_2 \cdot b_j =
p_2 \cdot \left(\sum_{i=1}^m\sum_{j=1}^n z_{ij} \cdot a_i \cdot w \cdot b_j\right) \cdot q_2 =
p_2 \cdot u \cdot q_2 = K.
$$
Note that this trick is possible because we are able to swap $a_i$ with $p_2$ and $q_2$ with $b_j$ in the expression above. So, to avoid this trick, we need to prevent the possibility of swapping.
One of the ways to do this is to change the expressions for $v$ and $u$ by using marginal tuples. We give three different schemes in Section~\ref{SecProto}. Our schemes differ from Roman'kov's because semigroups lack multiplicative inverses, and semirings lack both multiplicative and additive inverses.

The remainder of this paper is structured into five parts. In Section~\ref{SecPrelim}, we recall some definitions from tropical algebra. Section~\ref{SecDef} contains our definitions of marginal sets in semigroups and semirings and some examples of marginal sets in tropical matrix algebras. In Section~\ref{SecProto}, we consider three abstract key exchange protocols with marginal sets based on semigroups and give examples of these protocols using tropical matrix algebras.
The final section offers a conclusion to the work.

\section{Preliminaries}\label{SecPrelim}
In this section, we recall the definitions of the max-plus semiring, the min-plus semiring, tropical matrix semirings, and related constructions.

The \textit{max-plus semiring} is the algebraic structure $\mathcal{R}_{{\max}, {+}}= \left<\mathbb{R} \cup \{-\infty\}, \oplus, \otimes\right>$, where the operations $\oplus$ and $\otimes$ are defined as follows:
\[
a \oplus b = \max(a, b) \quad \text{and} \quad a \otimes b = a + b
\]
for all $a, b \in \mathbb{R} \cup \{-\infty\}$. The additive identity is $-\infty$, and the multiplicative identity is $0$. 

The \textit{min-plus semiring} is the algebraic structure $\mathcal{R}_{{\min}, {+}} = \left<\mathbb{R} \cup \{+\infty\}, \oplus, \otimes\right>$ with the operations $\oplus$ and $\otimes$ defined as follows:
\[
a \oplus b = \min(a, b) \quad \text{and} \quad a \otimes b = a + b
\]
for all $a, b \in \mathbb{R} \cup \{+\infty\}$.
Here, the additive identity is $+\infty$, and the multiplicative identity remains $0$. 

Let $\mathcal{S} = \left<S, \otimes, \oplus\right>$  be a semiring. The set of $n \times n$ matrices over $\mathcal{S}$, denoted $\mathrm{Mat}(n, {S})$, can be equipped with matrix addition and multiplication defined as follows:
\[
(A \oplus B)_{ij} = a_{ij} \oplus b_{ij}, \qquad (A \otimes B)_{ij} = \bigoplus_{k=1}^n a_{ik} \otimes b_{kj}
\]
for $A = (a_{ij})$ and $B = (b_{ij})$ in $\mathrm{Mat}(n, S)$.
The obtained structure is an idempotent semiring as well.
We denote the semiring of matrices of size $k \times k$ over $\mathcal{R}_{{\min}, {+}}$ by $\mathcal{M}_{{\min}, {+}}^k$ and 
the semiring of matrices of size $k \times k$ over $\mathcal{R}_{{\max}, {+}}$ by $\mathcal{M}_{{\max}, {+}}^k$.

In this paper, we use capital letters for matrices and the corresponding lowercase letters for elements of matrices: $x_{ij}$ is an element of $X$.

Let $\mathcal{S} = \left<S, \otimes, \oplus\right>$ be a tropical semiring. A \textit{tropical polynomial} over $\mathcal{S}$ in one variable $x$ is an expression of the form
\[
p(x) = a_0 \oplus a_1 \otimes x \oplus a_2 \otimes x^{\otimes 2} \oplus \dots \oplus a_k \otimes x^{\otimes k},
\]
where $a_i \in S$, and $x^{\otimes i}$ denotes $\underbrace{x \otimes \dots \otimes x}_{i \text{ times}}$. 

Let $A \in \mathrm{Mat}(n, {S})$ be a matrix and $p(x) \in \mathcal{S}[x]$ be a polynomial over $\mathcal{S}$. Define $p(A)$ as follows:
\[p(A) = a_0 \otimes I \oplus a_1 \otimes A \oplus a_2 \otimes A^{\otimes 2} \oplus \dots \oplus a_k \otimes A^{\otimes k},
\]
where $I$ is the tropical identity matrix of size $n \times n$.

It is easy to see that for any two polynomials $p(x)$ and $q(x)$, the matrices $p(A)$ and $q(A)$ commute:
$$
p(A) \otimes q(A) = q(A) \otimes p(A). 
$$

For different implementations of Protocol~\ref{ProtocolSidelnikov}, the following types of tropical matrices were considered.

An $n \times n$ matrix is called \textit{circulant} if each row is a cyclic right shift of the previous row: 
\begin{equation}\label{CircMat} 
	\begin{pmatrix}
		c_{1}& c_{2}&\cdots & c_{n}\\
		c_{n} &c_{1}& \cdots &c_{n-1}\\
		\vdots &  & \ddots & \vdots \\
		c_{2} &c_{3}& \cdots & c_{1}\\
	\end{pmatrix}
\end{equation}
The matrix~(\ref{CircMat}) is denoted by $ C(c_1, \ldots, c_n)$.
Since any circulant matrix is a matrix polynomial in the cyclic permutation matrix $C(o, \ldots, o, e)$, any two circulant matrices commute.

A matrix of the form
	\[
	\begin{pmatrix}
		c_{1} & t \otimes c_{2} & t \otimes c_{3} & \cdots & t \otimes c_{n} \\
		c_{n} & c_{1} & t \otimes c_{2} & \cdots & t \otimes c_{n-1} \\
		c_{n-1} & c_{n} & c_{1} & \cdots & t \otimes c_{n-2} \\
		\vdots & \vdots & \ddots & \ddots & \vdots \\
		c_{2} & c_{3} & \cdots & c_{n} & c_{1} \\
	\end{pmatrix}
	\]
 is called \textit{upper-$t$-circulant}. 
 
Analogously, a matrix of the form
	\[
	\begin{pmatrix}
		c_{1} & c_{2} & c_{3} & \cdots & c_{n} \\
		s \otimes c_{n} & c_{1} & c_{2} & \cdots & c_{n-1} \\
		s \otimes c_{n-1} & s \otimes c_{n} & c_{1} & \cdots & c_{n-2} \\
		\vdots & \vdots & \ddots & \ddots & \vdots \\
		s \otimes c_{2} & s \otimes c_{3} & \cdots & s \otimes c_{n} & c_{1} \\
        \end{pmatrix}
	\]
is called \textit{lower-$s$-circulant}.

Any two upper-$t$-circulant matrices and any two lower-$s$-circulant matrices commute~\cite{Amutha, Huang,BuchinskiyKotovTreier2024A}.

The following type of matrices was considered by D. Jones~\cite{Jones2017}. A matrix $A = (a_{ij}) \in \mathrm{Mat}(n, \mathcal{R}_{{\max}, {+}})$ is called a \textit{Jones matrix} if it satisfies the following condition:
\[
	a_{ij} \oplus a_{jk} \leq a_{ik} \oplus a_{jj}, \quad \forall i,j,k \in \{1, \ldots, n\}.
\]

Let $A = (a_{ij})$ be a Jones matrix and $\alpha \in \mathbb{R}$. The matrix $A^{(\alpha)} = (b_{ij})$ defined by
$$
b_{ij} = a_{ij} \otimes (a_{ii} \oplus a_{jj})^{\otimes(\alpha - 1)}
$$
is called a \textit{deformation} of $A$.
Muanalifah and Sergeev proved~\cite{MuanalifahSergeev2020} that
if $A$ is a Jones matrix, then $$A^{(\alpha)} \otimes B^{(\beta)} = A^{(\beta)} \otimes A^{(\alpha)}$$ for any $\alpha$ and $\beta$ such that $0 \leq \alpha \leq 1$ and $0 \leq \beta \leq 1$.

Let $r \geq 0$ and $k \leq 0$ be real numbers. A matrix $A \in \mathrm{Mat}(n, \mathcal{R}_{\min, {+}})$ such that
	$a_{ii} = k$  and $a_{ij} \in [r, 2r]$, $i \neq j$,
is called a \textit{Linde--de la Puente matrix}~\cite{LindedelaPuente2015, MuanalifahSergeev2020}.
The set of such matrices is denoted by $[r, 2r]^k_n$. 
It was proved~\cite{MuanalifahSergeev2020} that for $A \in [r, 2r]^{k_1}_n$ and $B \in [s, 2s]^{k_2}_n$ such that $r, s \geq 0$ and $k_1, k_2 \leq 0$,
$A \otimes B = B \otimes A$. The definition for Linde--de la Puente matrices over $\mathcal{R}_{{\max}, +}$ is similar.

\section{Marginal sets in semigroups and semirings}\label{SecDef}

In this section, we introduce the concept of marginal sets for semigroups and semirings. We give algorithms for constructing some marginal sets for tropical matrix algebras. 

In the following definition, $\square_1, \ldots, \square_n$ are variables; we use the symbol $\square$ instead of a letter to highlight their special role.

\begin{definition}\label{def1}
Let $\mathcal{S} = \left<S, {\cdot}, e\right>$ be a semigroup with an identity element $e$, let $(a_1, \ldots, a_m) \in S^m$, and let $w(x_1, \ldots, x_m, \square_1, \ldots, \square_n)$ be a semigroup word. A~tuple $(c_1, \ldots, c_n) \in S^n$ is called $w(a_1, \ldots, a_m, \square_1, \ldots, \square_n)$-\textit{marginal}  if
\begin{equation}\label{MainEq}
w(a_1, \ldots, a_m, c_1, \ldots, c_n) = 
w(a_1, \ldots, a_m, e, \ldots, e).
\end{equation}
We denote this by writing $(c_1, \ldots, c_n) \perp w(a_1, \ldots, a_m, \square_1, \ldots, \square_n)$.
Any subset $C \subseteq S^n$ of $w(a_1, \ldots, a_m, \allowbreak \square_1, \ldots, \square_n)$-marginal tuples is called a $w(a_1, \ldots, a_m, \square_1, \ldots, \square_n)$-\textit{marginal} set.
We denote this by writing $C \perp w(a_1, \ldots, a_m, \square_1, \ldots, \square_n)$.
\end{definition}

\begin{example}
Consider the semigroup of tropical matrices of size $3 \times 3$ over the min-plus algebra. Let $w(x, \square) = x \otimes \square$ and 
$$A = \!\begin{pmatrix}
3 & 7 & 4 \\ 
5 & 12 & 7 \\
6 & 5 & 11 
\end{pmatrix}\!.
$$
It is easy to check that the set 
$$C
= 
\!\left\{\!\!\begin{pmatrix}
0 & 7 & 5 \\
1 & 0 & 6 \\
-1 & 5 & 0
\end{pmatrix}\!,
\begin{pmatrix}
0 & 7 & 5 \\
1 & 0 & 6\\
0 & 5 & 0
\end{pmatrix}\!\!
\right\}
$$
is $(A \otimes \square)$-marginal.
\end{example}

Note that marginal sets are subsets of the solution sets of the corresponding equations~(\ref{MainEq}).  For example, if $a \in S$, then $(c) \perp a \cdot \square$ for any solution $c$ to the equation $ax = a$.

The following lemmas describe the solutions to the following equations over $\mathcal{M}_{\min, +}^k$:
$A \otimes X = A$, $X \otimes A = A$, $X \otimes A \otimes Y = A$, and $A \otimes X \otimes B \otimes Y \otimes C =  A \otimes B \otimes C$. Note that it is possible to prove similar lemmas for $\mathcal{M}_{\max, +}^k$.

\begin{lemma}\label{lemma1}
Let $A \in \mathrm{Mat}(k, \mathbb{R})$. A matrix $X \in \mathrm{Mat}(k, \mathbb{R})$ is a solution to $A \otimes X = A$ over $\mathcal{M}_{\min, +}^k$ if and only if $X \geq X^*(A)$ and
$$\bigcup\limits_{(i,j) \; : \; x_{ij} = x^{*}_{i, j}(A)} M_{ij}(A) = \{1, \ldots, k\} \times \{1, \ldots, k\},$$
	where 
    $x^{*}_{ij}(A) = \max_l (a_{lj} \otimes a_{li}^{\otimes  -1})$ and 
    $M_{ij}(A) = \{(l, j) \; : \; x^{*}_{ij}(A) = a_{lj} \otimes a_{li}^{\otimes -1}\}$.
\end{lemma}	

\begin{proof}
Note that the equation $A \otimes X = A$ can be rewritten as a one-sided system of linear equations:
$$
    \begin{array}{lll}
        a_{11} \otimes x_{11} \oplus a_{12} \otimes x_{21} \oplus \ldots \oplus a_{1k} \otimes x_{k1}               & = & a_{11}, \\
        \qquad a_{11} \otimes x_{12} \oplus a_{12} \otimes x_{22} \oplus \ldots \oplus a_{1k} \otimes x_{k2}        & = & a_{12}, \\
        \hdotsfor{3} \\
        \qquad\qquad a_{11} \otimes x_{1k} \oplus a_{12} \otimes x_{2k} \oplus \ldots \oplus a_{1k} \otimes x_{kk}  & = & a_{1k}, \\
        a_{21} \otimes x_{11} \oplus a_{22} \otimes x_{21} \oplus \ldots \oplus a_{2k} \otimes x_{k1}               & = & a_{21}, \\
        \qquad a_{21} \otimes x_{12} \oplus a_{22} \otimes x_{22} \oplus \ldots \oplus a_{2k} \otimes x_{k2}        & = & a_{22}, \\
        \hdotsfor{3} \\
        \qquad\qquad a_{21} \otimes x_{1k} \oplus a_{22} \otimes x_{2k} \oplus \ldots \oplus a_{2k} \otimes x_{kk}  & = & a_{2k}, \\
        \hdotsfor{3}\\
        a_{k1} \otimes x_{11} \oplus a_{n2} \otimes x_{21} \oplus \ldots \oplus a_{kk} \otimes x_{k1}               & = & a_{k1}, \\
        \qquad a_{k1} \otimes x_{12} \oplus a_{k2} \otimes x_{22} \oplus \ldots \oplus a_{kk} \otimes x_{k2}        & = & a_{k2}, \\
        \hdotsfor{3} \\
        \qquad\qquad a_{k1} \otimes x_{1k} \oplus a_{k2} \otimes x_{2k} \oplus \ldots \oplus a_{kk} \otimes x_{kk}  & = & a_{kk}.
    \end{array}
$$
Applying Theorem 3.1.1 from \cite{Butkovic2010} to this system of equations, we obtain the desired result.  
\end{proof}

For $X \otimes A = A$, we have a similar lemma.

\begin{lemma}\label{lemma2}
Let $A \in \mathrm{Mat}(k, \mathbb{R})$. A matrix $X \in \mathrm{Mat}(k, \mathbb{R})$ is a solution to $X \otimes A = A$ over $\mathcal{M}_{\min, +}^k$ if and only if $X \geq X^*(A)$ and
$$\bigcup\limits_{(i,j) \; : \; x_{ij} = x^{*}_{i, j}(A)} M_{ij}(A) = \{1, \ldots, k\} \times \{1, \ldots, k\},$$
	where 
    $x_{ij}^*(A) = \max_l (a_{il} \otimes a_{jl}^{\otimes -1})$ and 
    $M_{ij}(A) = \{(i, l) \; : \; x_{ij}^*(A) = a_{il} \otimes a_{jl}^{\otimes -1}\}$.
\end{lemma}

\begin{proof}
The proof is similar to the proof of Lemma~\ref{lemma1}.
\end{proof}

There are different ways to construct marginal sets. Here, we suggest algorithms for generating marginal sets using Lemmas~\ref{lemma1} and \ref{lemma2}.

Note that 
$T^{*}= \{M_{11}, M_{22}, \ldots, M_{kk}\}$ forms a cover of $\{1, \ldots, k\} \times \{1, \ldots, k\}$. Indeed, 
$$x_{jj}^{*} = \max_l (a_{lj} - a_{lj}) = \max_l (0) = 0$$
and
$$M_{jj} = \{(l, j) \; : \; x_{jj}^* = a_{lj} - a_{lj}\} = 
\{(l, j) \; : \; 0 = 0\} = \{(l, j)\}_l.$$
Thus, 
$$\bigcup_{j=1}^{k} M_{jj} = \bigcup_{j=1}^k \{(l, j)\}_{1 \leq l \leq k} = \{1, \ldots, k\} \times \{1, \ldots, k\}.$$
Therefore, we can consider the following Algorithm~\ref{AlgoAXeqA2}.
\begin{algorithm}
\caption{Generate a set $C$ such that $C \perp A \otimes \square$} 
\label{AlgoAXeqA2}
\begin{algorithmic}
\State \textbf{Input:} A matrix $A \in \mathrm{Mat}(k, \mathbb{Z})$, an integer $n$, an integer $l$
\State \textbf{Output:} A set $C$ such that $C \perp A \otimes \square$ and $|C| = n$
\State $C \gets \emptyset$
\State $X^{*} \gets  (\max_p (a_{pj} - a_{pi}))_{1 \leq i \leq k, 1 \leq j \leq k}$
   \State $\hat{X} \gets \Call{GetMaxPossibleMatrix}{T^{*}, X^{*}, l}$
\While{$|C| < N$}
     \State Randomly select $X$ such that $X^{*} \leq X \leq \hat{X}$
    \State $C \gets C \cup \{X\}$
\EndWhile
\State \Return $C$
\end{algorithmic}
\end{algorithm}

The function $\Call{GetMaxPossibleMatrix}{T, X^{*}, l}$ returns the matrix 
$P = (p_{ij})$, where
$$
p_{ij} = 
\left\{
\begin{array}{ll}
x^{*}_{ij} & \text{if } (i, j) \in T, \\
\max(l, x^{*}_{ij}) & \text{if } (i, j) \notin T.
\end{array}
\right. 
$$

Note that the loop in Algorithm~\ref{AlgoAXeqA2} and the following algorithms is potentially infinite. To avoid that, we can add a counter and break the loop after some number of iterations.

\begin{example} Let $l = 100$, $n = 3$, and
$$
A =
\begin{pmatrix}
0 & 85 & -6 \\
-72 & 53 & -97 \\
-72 & 52 & -69
\end{pmatrix}\!.
$$
Then, 
$$
X^{*} = 
\begin{pmatrix}
0 & 125 & 3 \\
-85 & 0 & -91 \\
25 & 150 & 0
\end{pmatrix}
\quad
\text{and}
\quad
\hat{X} = 
\begin{pmatrix}
0 & 125 & 100 \\
100 &  0 &  100 \\
100 &  150 & 0
\end{pmatrix}\!.
$$
The algorithm can randomly generate the following three matrices between $X^{*}$ and $\hat{X}$:
$$
C =
\left\{\!
\begin{pmatrix}
0 & 125 & 65 \\
29 & 0 & -51 \\
61 & 150 & 0 
\end{pmatrix}\!,
\begin{pmatrix}
0 & 125 &  14 \\
29 & 0 & -91 \\
88 & 150 & 0 
\end{pmatrix}\!,
\begin{pmatrix}
0 & 125 & 76 \\
20 & 0 & -68 \\
71 & 150 & 0
\end{pmatrix}\!
\right\}.
$$
\end{example}

Similarly, using Lemma~\ref{lemma2}, we can write an algorithm to find a marginal set $C$ such that $C \perp \square \otimes A$; see Algorithm~\ref{AlgoXAeqA2}.

\begin{algorithm}
\caption{Generate a set $C$ such that $C \perp \square \otimes A$} 
\label{AlgoXAeqA2}
\begin{algorithmic}
\State \textbf{Input:} A matrix $A$, a number $n$, a number $l$
\State \textbf{Output:} A set $C$ such that $C \perp \square \otimes A$ and $|C| = N$
\State $C \gets \emptyset$
\State $X^{*} \gets  (\max \{a_{ip} - a_{jp}\}_p)_{1 \leq i \leq k, 1 \leq j \leq k}$
   \State $\hat{X} \gets \Call{GetMaxPossibleMatrix}{T^{*}, X^{*}, l}$

\While{$|C| < n$}
     \State Randomly select $X$ such that $X^{*} \leq X \leq \hat{X}$
    \State $C \gets C \cup \{X\}$
\EndWhile
\State \Return $C$
\end{algorithmic}
\end{algorithm}

\begin{remark}
Note that different compression techniques can be used to transfer marginal sets.
For example, if
$$
C = 
\left\{\!
\begin{pmatrix}
2 & 3 \\
4 & 5
\end{pmatrix}\!,
\begin{pmatrix}
2 & 4 \\
4 & 5
\end{pmatrix}\!,
\begin{pmatrix}
2 & 5 \\
4 & 5
\end{pmatrix}\!,
\begin{pmatrix}
2 & 6 \\
4 & 5
\end{pmatrix}\!,
\begin{pmatrix}
2 & 7 \\
4 & 5
\end{pmatrix}\!,
\begin{pmatrix}
2 & 3 \\
5 & 5
\end{pmatrix}\!,
\begin{pmatrix}
2 & 4 \\
5 & 5
\end{pmatrix}\!,
\begin{pmatrix}
2 & 5 \\
5 & 5
\end{pmatrix}\!,
\begin{pmatrix}
2 & 6 \\
5 & 5
\end{pmatrix}\!,
\begin{pmatrix}
2 & 7 \\
5 & 5
\end{pmatrix}\!\right\}\!,
$$
then the following encoding can be transferred:
$$
\begin{pmatrix}
2 & [3, 7] \\
[4, 5] & 5
\end{pmatrix}\!,
$$
where $[a, b]$ means any integer between $a$ and $b$. Another technique is to send the first matrix and the difference between two consecutive matrices. For example, if 
$$
C = 
\left\{\!
\begin{pmatrix}
2 & 3 & 4 \\
4 & 5 & 1 \\
0 & 8 & 6
\end{pmatrix}\!,
\begin{pmatrix}
2 & 3 & 7 \\
4 & 5 & 1 \\
0 & 8 & 6
\end{pmatrix}\!,
\begin{pmatrix}
2 & 3 & 8 \\
4 & 5 & 2 \\
0 & 8 & 6
\end{pmatrix}\!\right\}\!,
$$
then the following encoding can be transferred:
$$
\begin{pmatrix}
2 & 3 & 4 \\
4 & 5 & 1 \\
0 & 8 & 6
\end{pmatrix}\!,
(((1, 3), 7)),
(((1, 3), 8), ((2, 3), 2)),
$$
where $((i, j), a)$ means that the element of the next matrix at position $(i,j)$ is $a$. 
\end{remark}

Now, consider equations of degree two. It is known that in general the problem to solve a system of tropical equations of degree two is $\mathsf{NP}$-hard~\cite{GrigorievShpilrain2014, BuchinskiyKotovTreier2024B}, but we will demonstrate that the system of equations in this section are simpler.

Note that a set $C = \{(M_1, M_2) : M_1 \in C_1, M_2 \in C_2\}$,
where $C_1 \perp \square \otimes A$ and $C_2 \perp A \otimes \square$ is $(\square_1 \otimes A \otimes \square_2)$-marginal. However, it is possible to construct more interesting $(\square_1 \otimes A \otimes \square_2)$-marginal sets.

\begin{lemma}\label{lemma3}
Let $A \in \mathrm{Mat}(k, \mathbb{R})$. A pair of matrices $(X, Y)$, $X, Y \in \mathrm{Mat}(k, \mathbb{R})$, is a solution to $X \otimes A \otimes Y = A$ over $\mathcal{M}_{\min, +}^k$ if and only if 
$x_{ip} \otimes y_{qj} \geq x_{ipqj}^{*}(A)$ and
for each $i$ and $j$, there exist $p$ and $q$ such that
\begin{equation}\label{Lemma3Cnd}
x_{iq} \otimes y_{pj} = x^{*}_{ipqj}(A),
\end{equation}
where 
$x_{ipqj}^*(A) = a_{ij} \otimes a_{pq}^{\otimes -1}$.
\end{lemma}	

\begin{proof}
Note that the equation can be rewritten in the following form:
$$ 
\bigoplus_{p = 1}^k \bigoplus_{q = 1}^k x_{ip} \otimes a_{pq} \otimes y_{qj} = a_{ij}. 
$$
Introduce new variables $t_{ipqj} = x_{ip} \otimes x_{qj}$. The system can be rewritten as
$$ 
\bigoplus_{p = 1}^k \bigoplus_{q = 1}^k (a_{pq} \otimes a_{ij}^{\otimes-1}) \otimes t_{ipqj} = 0. 
$$
To solve this system of linear equations, we can again apply Theorem 3.1.1 from \cite{Butkovic2010}. Since the variable $t_{ipqj}$ is presented only in one equation, 
$x_{ipqj}^*(A) = a_{ij} \otimes a_{pq}^{\otimes -1}$. The requirement that at least one term must be equal to 0 gives up the requirement~(\ref{Lemma3Cnd}).
\end{proof}

To apply this lemma to construct a marginal set $C$ such that $C \perp \square_1 \otimes A \otimes \square_2$, where $A \in \mathrm{Mat}(k, \mathbb{Z})$, consider the following linear programming problem. 
First, add the inequalities
$$
x_{ip} \otimes y_{qj} \geq x_{ipqj}^{*}(A)
$$
to the set of constraints. 
Second, for each $i$ and $j$ add the equation 
$$
x_{ii} \otimes y_{jj} = 0$$
to the set of constraints.
To add more randomness, 
select a random integer $d$, random integers $r_{ij}$, $i \neq j$, and $s_{ij}$, $i \neq j$, and add inequalities 
$$
\begin{array}{ll}
x_{ij} \geq r_{ij}, & 1 \leq i \leq k, 1 \leq j \leq k, i \neq j, \\
y_{ij} \geq s_{ij}, & 1 \leq i \leq k, 1 \leq k \leq k, i \neq j, \\
x_{ii} \geq d, & 1 \leq i \leq k, \\
y_{jj} \geq -d, & 1 \leq j \leq k, \\
\end{array}
$$
to the set of constraints as well. 

Note that it follows from Ghouila-Houri's Theorem ~\cite{GhouilaHouri1962} that the matrix of this linear programming problem is totally unimodular. Therefore, from the Hoffman--Kruskal Theorem\cite{HoffmanKruskal1956} it follows that the problem is integral, i.e. it has an integral optimum whenever an optimum exists. So, to solve this problem, we can use the simplex method or another method to solve a linear program. This gives us Algorithm~\ref{AlgoXAXeqA}. Note that each iteration of the algorithm runs in polynomial time.

\begin{algorithm}
\caption{Generate a set $C$ such that $C \perp \square_1 \otimes A \otimes \square_2$} 
\label{AlgoXAXeqA}
\begin{algorithmic}
\State \textbf{Input:} A matrix $A$, a number $n$, a number $l_1$, and a number $l_2$
\State \textbf{Output:} A set $C$ such that $C \perp \square_1  \otimes A \otimes \square_2$ and $|C| = n$
\State $C \gets \emptyset$
\While{$|C| < n$}
     \State $constraints \gets \{x_{ip} + y_{qj} \geq a_{ij} - a_{pq}\}_{1 \leq i,p,q,j \leq k}$
     \State $constraints \gets constraints \cup \{x_{ii} + y_{jj} = 0\}_{1 \leq i, j \leq k}$
    \State $d \gets \Call{GetRandom}{\{l_1, \ldots, l_2\}}$
     \For{$1 \leq i, j \leq k$}
     \If{$i \neq j$}
    \State $r_{ij} \gets \Call{GetRandom}{\{l_1, \ldots, l_2\}}$

    \State $s_{ij} \gets \Call{GetRandom}{\{l_1, \ldots, l_2\}}$
    \Else
    \State $r_{ij} \gets d$
    \State $s_{ij} \gets -d$
    \EndIf
    \EndFor
    \State $constraints \gets constraints \cup \{x_{ij} \geq r_{ij}\}_{ij}$
    \State $constraints \gets constraints \cup \{y_{ij} \geq s_{ij}\}_{ij}$
     \State $objective \gets 0$
     \State $(X, Y) \gets \Call{SolveLinearProgram}{constraints, objective}$ 
     \State $C \gets C \cup \{(X, Y)\}$
\EndWhile
\State \Return $C$
\end{algorithmic}
\end{algorithm}

\begin{example}
Consider how Algorithm~\ref{AlgoXAXeqA} generates one tuple.
Let $A = 
\begin{pmatrix}
3 & 2 \\ 
1 &5
\end{pmatrix}
$. Then the constraints are:
\begingroup
\allowdisplaybreaks
\begin{align*}
x_{11} + y_{11} &= 0, \\
x_{11} + y_{12} &\geq -1, \\
x_{11} + y_{21} &\geq 1, \\
x_{11} + y_{22} &= 0, \\
x_{12} + y_{11} &\geq 2, \\
x_{12} + y_{12} &\geq 1, \\
x_{12} + y_{21} &\geq -2, \\
x_{12} + y_{22} &\geq -3, \\
x_{21} + y_{11} &\geq -2, \\ 
x_{21} + y_{12} &\geq 2, \\
x_{21} + y_{21} &\geq -1, \ \\
x_{21} + y_{22} &\geq 3, \\
x_{22} + y_{11} &= 0, \\
x_{22} + y_{12} &\geq 4,  \\
x_{22} + y_{21} &\geq -4,  \\
x_{22} + y_{22} &= 0  
\end{align*}
We can randomly select the following lower bounds:
\begin{align*}
x_{11} &\geq 3, \\
x_{12} &\geq 9, \\
x_{21} &\geq 7, \\
x_{22} &\geq 3, \\
y_{11} &\geq -3, \\
y_{12} &\geq 4, \\
y_{21} &\geq 0, \\
y_{22} &\geq -3.
\end{align*}
\endgroup

The simplex method gives the following matrices:
$$
X =
\begin{pmatrix}
3 & 9 \\
7 & 3     
\end{pmatrix}
\quad
\text{and}
\quad
Y = 
\begin{pmatrix}
-3 & 4 \\ 0 & -3
\end{pmatrix}.
$$

It is easy to see that
$X \otimes A = 
\begin{pmatrix}
5 & 4 \\ 3 & 7
\end{pmatrix}
\neq A
$
and
$A \otimes Y =
\begin{pmatrix}
1 & 0 \\ -1 & 3
\end{pmatrix}
\neq A
$
but
$
X \otimes A \otimes Y =  
\begin{pmatrix}
3 & 2 \\ 1 &  5
\end{pmatrix}
= A. 
$
\end{example}

\begin{lemma}\label{lemma4}
Let $A, B, C \in \mathrm{Mat}(k, \mathbb{R})$. A pair of matrices $(X, Y)$, $X, Y \in \mathrm{Mat}(k, \mathbb{R})$, is a solution to 
$$A \otimes X \otimes B \otimes Y \otimes C = A \otimes B \otimes C$$ over $\mathcal{M}_{\min, +}^k$ if and only if 
$x_{pq} \otimes y_{rs} \geq x_{pqrs}^{*}(A, B, C)$ and
$$\bigcup\limits_{(p,q,r,s) \; : \; x_{pq} \otimes y_{rs} = x^{*}_{pqrs}(A, B, C)} M_{pqrs}(A, B, C) = \{1, \ldots, k\} \times \{1, \ldots, k\},$$
	where 
    $$x_{pqrs}^*(A, B, C) = 
    \max \{
    d_{ij} \otimes
    a_{ip}^{\otimes -1} \otimes b_{qr}^{\otimes -1} \otimes c_{sj}^{\otimes -1}
    \}_{i,j},$$
    $$d_{ij} = 
\bigoplus_{u = 1}^k \bigoplus_{v = 1}^k
a_{iu} \otimes b_{uv} \otimes c_{vj},$$
and
    $$M_{p q r s}(A, B, C) = \{(i, j) \; : \; x_{pqrs}^*(A) = 
    d_{ij} \otimes
    a_{ip}^{\otimes -1} \otimes b_{qr}^{\otimes -1} \otimes c_{sj}^{\otimes -1}
    \}.$$
\end{lemma}	

\begin{proof}
Note that the equation can be rewritten in the following form:
$$ 
\bigoplus_{p = 1}^k \bigoplus_{q = 1}^k
\bigoplus_{r = 1}^k \bigoplus_{s = 1}^k
a_{ip} \otimes x_{pq} \otimes b_{qr} \otimes y_{rs} \otimes c_{sj} = 
\bigoplus_{u = 1}^k \bigoplus_{v = 1}^k
a_{iu} \otimes b_{uv} \otimes c_{vj}.
$$
Introduce new variables $t_{pqrs} = x_{pq} \otimes x_{rs}$. The system can be rewritten as
$$ 
\bigoplus_{p = 1}^k \bigoplus_{q = 1}^k
\bigoplus_{r = 1}^k \bigoplus_{s = 1}^k
(a_{ip} \otimes b_{qr} \otimes c_{sj} \otimes d_{ij}^{\otimes -1}) \otimes t_{pqrs} 
= 0.
$$
Applying Theorem 3.1.1 from \cite{Butkovic2010} to this system of equations completes the proof.
\end{proof}

Now, let us discuss how this lemma can we use to create an algorithm to generate marginal sets $M$, $M\perp A \otimes \square_1 \otimes B \otimes \square_2 \otimes C$, where $A \in \mathrm{Mat}(k, \mathbb{Z})$.

Note that
$$\bigcup\limits_{(p,r) \; : \; x^{*}_{pprr}(A, B, C) = 0} M_{pprr}(A) = \{1, \ldots, k\} \times \{1, \ldots, k\}.$$
To prove it, first, note that $x^{*}_{pprr}(A, B, C) \leq 0$.  
Indeed, 
$x^{*}_{pprr}(A, B, C) =
    \max_{i,j} (
    d_{ij} \otimes
    a_{ip}^{\otimes -1} \otimes b_{pr}^{\otimes -1} \otimes c_{rj}^{\otimes -1}
    ) = 
\max_{i,j} (
    \min_{u, v}(a_{iu} + b_{uv} + c_{vj}) - 
    (a_{ip} + b_{pr} + c_{rj})
    ) \leq 0$ because $
\min_{u, v}(a_{iu} + b_{uv} + c_{vj}) \leq 
    a_{ip} + b_{pr} + c_{rj}$.
Second, to find $p = p(i, j), q = q(i, j)$ for each $(i, j)$,
take any $(p, q) \in \mathrm{argmin}_{u, v}(a_{iu} + b_{uv} + c_{vj})$.

Therefore, consider the following linear programming problem. First, add the inequalities
$$
x_{pq} \otimes y_{rs} \geq x_{pqrs}^{*}(A, B, C)
$$
to the set of constraints. Second, 
let 
\begin{align*}    
P(A, B, C) &= \{(p, r) : x^*_{pprr}(A, B, C) = 0\}, \\
P_x(A, B, C) &= \{p : \exists\, r \text{ such that } x^*_{pprr}(A, B, C) = 0\}, \\
P_y(A, B, C) &= \{r : \exists\, p \text{ such that } x^*_{pprr}(A, B, C) = 0\}.
\end{align*}
For each $(p, r) \in P(A, B, C)$ add the equations
$$
x_{pp} \otimes y_{rr} = 0
$$
to the set of constraints.
Third, to add more randomness, 
select a random integer $h$ and integers $r_{ij}, s_{ij}$, 
such that
$$
\begin{array}{ll}
r_{ii} = h & \text{if } i \in P_x(A, B, C), \\
r_{ij} \text{ is a random integer} & \text{otherwise},
\end{array}
$$
and 
$$
\begin{array}{ll}
s_{jj} = -h & \text{if } j \in P_y(A, B, C), \\
s_{ij} \text{ is a random integer} & \text{otherwise},
\end{array}
$$
and add inequalities 

$$
\begin{array}{ll}
x_{ij} \geq r_{ij}, 1 \leq i \leq k, 1 \leq j \leq k,\\
y_{ij} \geq s_{ij}, 1 \leq i \leq k, 1 \leq j \leq k.
\end{array}
$$
to the set of constraints as well.  Similarly, to solve this problem, we can apply
the simplex method or another method for linear programming.

\begin{algorithm}
\caption{Generate a set $C$ such that $C \perp A \otimes \square_1 \otimes B \otimes \square_2 \otimes C$} 
\label{AlgoAXBYCeqABC}
\begin{algorithmic}
\State \textbf{Input:} Matrices $A, B$, and $C$, a number $n$, numbers $l_1$ and $l_2$
\State \textbf{Output:} A set $F$ such that $F \perp \square_1  \otimes A \otimes \square_2$ and $|D| = n$
\State $D = (\min_{u, v}(a_{iu} + b_{uv} + c_{vj}))_{1 \leq i \leq k, 1 \leq j \leq k}$
\State $X^* \gets (\max_{i, j}(d_{ij} - a_{ip} - b_{qr} - c_{sj}))_{1 \leq p \leq k, 1 \leq q \leq k, 1 \leq r \leq k, 1 \leq s \leq k}$
\State $P \gets \{(p, r) : x^*_{pprr} = 0\}$
\State $P_x \gets  \{p : \exists\, r \text{ such that } x^*_{pprr} = 0\}$
\State $P_y \gets \{r : \exists\, p \text{ such that } x^*_{pprr} = 0\}$
\State $F \gets \emptyset$
\While{$|F| < n$}
     \State $constraints \gets \{x_{pq} + y_{rs} \geq x_{pqrs}\}_{1 \leq p,q,r,s \leq k}$
     \State $constraints \gets constraints \cup \{x_{pp} + y_{qq} = 0\}_{(p, q) \in P}$
    \State $h \gets \Call{GetRandom}{\{l_1, \ldots, l_2\}}$
     \For{$1 \leq i, j \leq k$}
     \If{$i = j \wedge i \in P_x$}
    \State $r_{ii} \gets h$
    \Else
    \State $r_{ij} \gets \Call{GetRandom}{\{l_1, \ldots, l_2\}}$
    \EndIf 
    \If{$i = j \wedge j \in P_y$}
    \State $s_{jj} \gets -h$
    \Else
    \State $s_{ij} \gets \Call{GetRandom}{\{l_1, \ldots, l_2\}}$
    \EndIf 
    \EndFor
    \State $constraints \gets constraints \cup \{x_{ij} \geq r_{ij}\}_{i,j}$
    \State $constraints \gets constraints \cup \{y_{ij} \geq s_{ij}\}_{i,j}$
     \State $objective \gets 0$
     \State $(X, Y) \gets \Call{SolveLinearProgram}{constraints, objective}$ 
     \State $F \gets F \cup \{(X, Y)\}$
\EndWhile
\State \Return $F$
\end{algorithmic}
\end{algorithm}

\begin{example} 
Consider the following example:
$$
A = \begin{pmatrix} -4 & 6 & 2 \\ -2 & -3 & 10 \\ -2 & -9 & -5\end{pmatrix}\!,
B = \begin{pmatrix} -4 & -10 & -3 \\ 2 & -2 & 8 \\ 4 & -1 & 6\end{pmatrix}\!, \text{ and }
C = \begin{pmatrix} -9 & 9 & -2 \\ 5 & -2 & -8 \\ 3 & 3 & 8\end{pmatrix}\!.
$$
Since $X^{*}(A, B, C)$ is a four-way table, we will write it as a block matrix, where the first two indexes specify the position of the block and the last two indexes specify the position of the element within that block.
For these matrices
$$
X^{*}(A, B, C) = 
\left(
\begin{array}{ccc|ccc|ccc}
0 & -6 & -11 & -6 & -12 & -17 & -8 & -14 & -19 \\
6 & 0 & -5  & -2 & -8  & -13 & -3 & -9  & -14 \\
-1 & -7 & -12 & -12 & -18 & -23 & -10 & -16 & -21 \\ \hline
6 & 1 & -4  & 0 & -5  & -10 & -2 & -7  & -12 \\
12 & 7 & 2  & 4 & -1  & -6  & 3 & -2  & -7  \\
5 & 0 & -5  & -6 & -11 & -16 & -4 & -9  & -14 \\ \hline
2 & -3 & -8  & -4 & -9  & -14 & -6 & -11 & -16 \\
8 & 3 & -2  & 0 & -5  & -10 & -1 & -6  & -11 \\
1 & -4 & -9  & -10 & -15 & -20 & -8 & -13 & -18
\end{array}
\right)\!,
$$
$P = \{(1, 1), (1, 2), (2, 1)\}$, $P_x = \{1, 2\}$, and $P_y = \{1, 2\}$.

Let $R$ and $S$ be randomly selected as
$$
R = 
\begin{pmatrix}
-10 & -4 & -1 \\
1 &-10 & 0 \\
6 & -1 & -2 
\end{pmatrix}
\text{ and }
S = 
\begin{pmatrix}
10 & 10  & 3 \\
-9 &10 & 0 \\
5 & 0 & 0
\end{pmatrix}\!.
$$

Then the set of constraints is
$$
\begin{array}{ll}
x_{pq} + y_{rs} \geq x_{pqrs}, & 1 \leq p, q, r, s  \leq 3 \\
x_{11} + y_{11} = 0, & \\
x_{11} + y_{22} = 0, & \\
x_{22} + y_{11} = 0, & \\
x_{ij} \geq r_{ij}, & 1 \leq i,j \leq 3 \\
y_{ij} \geq s_{ij}, & 1 \leq i, j \leq 3.
\end{array}
$$
Note that $r_{11} = r_{22} = 10$ and $s_{11} = s_{22} = -10$.
Solving the linear programming problem, we obtain
$$
X = \begin{pmatrix}-10 & -4 & -1 \\ 1 & -10 & 0 \\ 6 & -1 & -2\end{pmatrix}
\text{ and }
Y = \begin{pmatrix}10 & 10 & 3 \\ 16 & 10 & 5 \\ 9 & 3 & 0\end{pmatrix}\!.
$$
Note that $A \otimes X \otimes B \otimes Y \otimes C = A \otimes B \otimes C$, but
$A \otimes X  \neq A$,
$X \otimes B \neq B$,
$B \otimes Y \neq B$,
$Y \otimes C \neq C$,
$A \otimes X \otimes B  \neq A \otimes B$,
$X \otimes B \otimes Y  \neq  B$, and
$B \otimes Y \otimes C \neq B \otimes C$.
\end{example}

Lemma~\ref{lemma3} and Lemma~\ref{lemma4} can be generalized to an arbitrary number of matrices. For example, consider how Lemma~\ref{lemma4} can be generalized.

\begin{lemma}\label{lemma5}
	Let $A_i\in \mathrm{Mat}(k, \mathbb{R})$, $1 \leq i \leq n + 1$. A tuple of matrices $(X_1, X_2, \ldots, X_n)$, $X_i \in \mathrm{Mat}(k, \mathbb{R})$, is a solution to 
	$$A_1 \otimes X_1 \otimes A_2 \otimes X_2 \ldots \otimes X_n \otimes A_{n + 1} = A_1 \otimes A_2 \otimes \ldots \otimes A_{n+1}$$ over $\mathcal{M}_{\min, +}^k$ if and only if $x_{1, p_1, q_1} \otimes \ldots \otimes x_{n, p_n, q_n} \geq x_{p_1, q_1, \ldots, p_n, q_n}^{*}(A_1, \ldots, A_n)$ and
	$$\bigcup\limits_{(p_1,q_1,\ldots, p_n,q_n) \; : \; x_{1, p_1, q_1} \otimes \ldots \otimes x_{n, p_n, q_n} = x_{p_1, q_1, \ldots, p_n, q_n}^{*}(A_1, \ldots, A_n)} M_{p_1, q_1, \ldots, p_n, q_n}(A_1, \ldots, A_n) = \{1, \ldots, k\} \times \{1, \ldots, k\},$$
	where 
	$$x_{p_1, q_1,\ldots, p_n, q_n}^{*}(A_1, \ldots, A_n) = 
	\max \{
	d_{i,j} \otimes
	a_{1, i, p_1}^{\otimes -1} \otimes a_{2, q_1, p_2}^{\otimes -1} \otimes \ldots \otimes a_{n+1, q_n, j}^{\otimes -1}
	\}_{i,j},$$
	$$d_{i,j} = 
	\bigoplus_{u_1 = 1}^k \bigoplus_{u_2 = 1}^k \ldots \bigoplus_{u_n = 1}^k
	a_{1, i, u_1} \otimes a_{2, u_1, u_2} \otimes \ldots \otimes a_{n+1, u_n, j},$$
	and
	$$M_{p_1, q_1, \ldots p_n, q_n}(A_1, \ldots, A_n) = \{(i, j) \; : \; x_{p_1, q_1, \ldots, p_n, q_n}^*(A_1, \ldots, A_n) = 
	d_{i,j} \otimes
	a_{1, i, p_1}^{\otimes -1} \otimes a_{2, q_1, p_2}^{\otimes -1} \otimes \ldots \otimes a_{n+1, q_n, j}^{\otimes -1}
	\}.$$
\end{lemma}

\begin{proof}
	Note that the equation can be rewritten in the following form:
	\begin{align*}
			\bigoplus_{p_1 = 1}^k \bigoplus_{q_1 = 1}^k
		\ldots \bigoplus_{p_n = 1}^k \bigoplus_{q_n = 1}^k
		a_{1,i,p_1} \otimes x_{1,p_1,q_1} \otimes \ldots \otimes x_{n,p_n,q_n} \otimes a_{n+1,q_n,j} = \\ 
		= \bigoplus_{u_1 = 1}^k \bigoplus_{u_2 = 1}^k \ldots \bigoplus_{u_n = 1}^k
		a_{1,i,u_1} \otimes a_{2,u_1,u_2} \otimes \ldots \otimes a_{n+1,u_n,j}.
	\end{align*}
	Introduce new variables $t_{p_1,q_1,\ldots,p_n,q_n} = x_{1,p_1,q_1} \otimes \ldots \otimes x_{n,p_n,q_n}$. The system can be rewritten as
	$$ 
	\bigoplus_{p_1 = 1}^k \bigoplus_{q_1 = 1}^k
	\ldots \bigoplus_{p_n = 1}^k \bigoplus_{q_n = 1}^k
	(a_{1,i,p_1} \otimes a_{2,q_1,p_2} \otimes \ldots \otimes a_{n+1,q_n,j} \otimes d_{i,j}^{\otimes -1}) \otimes t_{p_1,q_1,\ldots,p_n,q_n}
	= 0.
	$$
	Applying Theorem 3.1.1 from \cite{Butkovic2010} to this system of equations completes the proof.
\end{proof}

The following definition is an attempt to define marginal sets in semirings. Again, $\square_1, \ldots, \square_n, \bigcirc_1, \ldots, \bigcirc_k$ are variables here. We use these symbols instead of letters to highlight their special role.

\begin{definition}
Let $\mathcal{R} = \left<R, {\cdot}, {+}, e, o\right>$ be a semiring, 
let $w(x_1, \ldots, x_m, \square_1, \ldots, \square_n, \bigcirc_1, \ldots, 
\bigcirc_k)$ be a semiring expression such that each $\square_i$, $1 \leq i \leq n$, is a factor and each $\bigcirc_j$, $1 \leq j \leq k$, is a summand, and let $(a_1, \ldots, a_m) \in R^m$.
A tuple $(c_1, \ldots, c_n, d_1, \ldots, d_k) \in R^{n+k}$ is called $w(a_1, \ldots, a_m, \square_1, \ldots, \square_n, \bigcirc_1, \ldots, \bigcirc_k)$-\textit{marginal}  if
$$
w(a_1, \ldots, a_m, c_1, \ldots, c_n, d_1, \ldots, d_k) = 
w(a_1, \ldots, a_m, \underbrace{e, \ldots, e}_n, \underbrace{o, \ldots, o}_k).
$$  
Any subset $C \subseteq R^{n+k}$ of $w(a_1, \ldots, a_m, \square_1, \ldots, \square_n, \bigcirc_1, \ldots, \bigcirc_k)$-\textit{marginal} tuples is called a
 $w(a_1, \ldots, a_m, \allowbreak \square_1, \ldots, \square_n, \bigcirc_1, \ldots, \bigcirc_k)$-\textit{marginal} set.
\end{definition}

\begin{example}
Let $$
A =
\begin{pmatrix}
3 & 7 & 4 \\
5 & 12 & 7 \\
6 & 5 & 11
\end{pmatrix}\!.
$$
$$
C = \left\{\!\!
\begin{pmatrix}
5 & 12 & 10 \\
25 & 12 & 8 \\
59 &  23 & 12
\end{pmatrix}\!\!
\right\}
$$ is a ($A \oplus \bigcirc)$-marginal set.
\end{example}

The following lemma helps construct ($A \oplus \bigcirc)$-marginal sets.

\begin{lemma}
Let $A \in Mat(k, \mathbb{R})$ A matrix $X \in \mathrm{Mat}(k, \mathbb{R})$ is a solution to $A \oplus X = A$ over $\mathcal{M}_{{\min}, {+}}^k$ if and only if $X \geq A$.
\end{lemma}
\begin{proof}
For any $i$ and $j$ we have $a_{ij} \otimes x_{ij} = a_{ij}$.
This means that $\min(a_{ij}, x_{ij}) = a_{ij}$. Therefore, 
$x_{ij} \geq a_{ij}$.
\end{proof}

\section{Key exchange protocols with marginal sets}\label{SecProto}

In this section, we consider four key-exchange schemes based on semigroups with marginal sets. These schemes use the following marginal sets:
$M \perp a \cdot \square$, $M \perp \square \cdot a$,
$M \perp \square_1 \cdot a \cdot \square_2$, and
$M \perp a \cdot \square_1 \cdot b \cdot \square_2 \cdot c$.
These schemes generalize the Sidelnikov--Cherepnev--Yaschenko protocol.
We use tropical matrix algebras to demonstrate how these schemes work. We do not state that these schemes are secure, but we use them as illustrations of how marginal sets can change protocols.

\begin{protocol}
Alice and Bob agree on a public noncommutative semigroup $\mathcal{S} = \left<S, \cdot, e\right>$, two public pairwise commuting subsets $H \subseteq S$ and $R \subseteq S$, and a public element $w \in S$.
\begin{enumerate}
\item Alice chooses two private elements $p_1 \in H$ and $q_1 
\in R$, constructs public marginal sets $M_1 \subseteq S$ and $N_1 \subseteq S$ such that $M_1 \perp p_1 \cdot \square$ and $N_1 \perp \square \cdot q_1$, and sends the sets $M_1$ and $N_1$ to Bob.

\item Bob chooses two private elements $p_2 \in H$ and $q_2 
\in R$, constructs public marginal sets $M_2 \subseteq S$ and $N_2 \subseteq S$ such that $M_2 \perp p_2 \cdot \square$ and $N_2 \perp \square \cdot q_2$, and sends the sets $M_2$ and $N_2$ to Alice.

\item  Alice chooses private elements $c_2 \in M_2$ and $d_2 \in N_2$,  
computes $u = c_2 \cdot p_1 \cdot w \cdot q_1 \cdot d_2$, and sends $u$ to Bob.

\item  Bob chooses private elements $c_1 \in M_1$ and $d_1 \in N_1$,  
computes $v = c_1 \cdot p_2 \cdot w \cdot q_2 \cdot d_1$, and sends $v$ to Alice.

\item Alice computes her shared secret key $k_A = p_1 \cdot v \cdot q_1$.

\item Bob computes the shared secret key $k_B = p_2 \cdot u \cdot q_2$.
\end{enumerate}
\end{protocol}

They share the same key because 
$$
k_A = p_1 \cdot v \cdot q_1 = 
p_1 \cdot (c_1 \cdot p_2 \cdot w \cdot q_2 \cdot d_1) \cdot q_1 = 
(p_1 \cdot c_1) \cdot p_2 \cdot w \cdot q_2 \cdot (d_1 \cdot q_1) = 
p_1 \cdot p_2 \cdot w \cdot q_2 \cdot q_1
$$
and
\begin{multline*}
k_B = p_2 \cdot u \cdot q_2 = 
p_2 \cdot (c_2 \cdot p_1 \cdot w \cdot q_1 \cdot d_2) \cdot q_2 = 
(p_2 \cdot c_2) \cdot p_1 \cdot w \cdot q_1 \cdot (d_2 \cdot q_2) = 
p_2 \cdot p_1 \cdot w \cdot q_1 \cdot q_2 = 
(p_2 \cdot p_1) \cdot w \cdot (q_1 \cdot q_2) = \\
p_1 \cdot p_2 \cdot w \cdot q_2 \cdot q_1. 
\end{multline*}

\begin{example}
Let us demonstrate how this general scheme can be applied to a tropical key-exchange protocol. Various options of pairwise commuting subsets described in Section~\ref{SecPrelim} can be used. In this example, let us consider the protocol proposed in~\cite{GrigorievShpilrain2014}, which uses polynomial of matrices.
Note that the original protocol and some of its modifications are well analyzed in~\cite{KotovUshakov2018, AlhussainiSergeev2024,alhussaini2024security,AlhussainiCollettSergeev2023}.

Alice and Bob agree on the semigroup of tropical matrices of size $3 \times 3$ over the min-plus algebra $\mathcal{M}_{\min, {+}}^3$ and three 
matrices
\[
A = \begin{pmatrix}
54 & 15 & 33 \\
59 & 87 & 53 \\
9 & 63 & 80
\end{pmatrix}\!,
B = \begin{pmatrix}
50 & 11 & 14 \\
16 & 29 & 33 \\
27 & 86 & 96
\end{pmatrix}\!, \text{ and } W = 
\begin{pmatrix}
54 & -67 & 35 \\
23 & -7 & -84 \\
9 & 97 & 33
\end{pmatrix}\!.\]

\begin{enumerate}
\item 
Alice randomly chooses two private elements
\[
P_1 = -69 \oplus -97 \otimes A \oplus 60 \otimes A^{\otimes2} = 
\begin{pmatrix}
-69 & -82 & -64 \\
-38 & -69 & -44 \\
-88 & -34 & -69
\end{pmatrix}
\]
and
\[
Q_1 = 8 \oplus -93 \otimes A \oplus 69 \otimes A^{\otimes2} = 
\begin{pmatrix}
-43 & -82 & -79 \\
-77 & -64 & -60 \\
-66 & -7 & 3
\end{pmatrix}\!.
\]
She constructs the following public marginal sets using Algorithms~\ref{AlgoAXeqA2} and~\ref{AlgoXAeqA2}:
\[
M_1 =
\left\{\!
\begin{pmatrix}
0 & 63 & 29 \\
64 & 0 & 133 \\
131 & 66 & 0
\end{pmatrix}\!,
\begin{pmatrix}
0 & 103 & 134 \\
33 & 0 & 134 \\
27 & 46 & 0
\end{pmatrix}\!,
\begin{pmatrix}
0 & 142 & 99 \\
74 & 0 & 43 \\
112 & 142 & 0
\end{pmatrix}\!
\right\}\]
and
\[
N_1 =
\left\{\!
\begin{pmatrix}
0 & 54 & 149 \\
38 & 0 & -6 \\
97 & 96 & 0
\end{pmatrix}\!,
\begin{pmatrix}
0 & 44 & 97 \\
38 & 0 & 33 \\
99 & 144 & 0
\end{pmatrix}\!,
\begin{pmatrix}
0 & 71 & 91 \\
21 & 0 & 77 \\
141 & 147 & 0
\end{pmatrix}\!
\right\}\!.
\]
She sends $M_1$ and $N_1$ to Bob.
\item
Bob randomly chooses two private elements
\[
P_2 = -11 \oplus 2 \otimes A \oplus -88 \otimes A^{\otimes2} = 
\begin{pmatrix}
-46 & -19 & -20 \\
-26 & -14 & 4 \\
-25 & -64 & -46
\end{pmatrix}
\]
and
\[
Q_2 = -33 \oplus 37 \otimes A \oplus -41 \otimes A^{\otimes2} = 
\begin{pmatrix}
-33 & -1 & 3 \\
4 & -33 & -11 \\
36 & -3 & -33
\end{pmatrix}\!.
\]
He constructs the following public marginal sets using Algorithms~\ref{AlgoAXeqA2} and~\ref{AlgoXAeqA2}:
\[
M_2 =
\left\{\!
\begin{pmatrix}
0 & 70 & 95 \\
92 & 0 & 29 \\
54 & 113 & 0
\end{pmatrix}\!,
\begin{pmatrix}
0 & 52 & 105 \\
110 & 0 & 30 \\
117 & 133 & 0
\end{pmatrix}\!,
\begin{pmatrix}
0 & 70 & 61 \\
119 & 0 & 20 \\
84 & 84 & 0
\end{pmatrix}\!
\right\}
\]
and
\[
N_2 =
\left\{\!
\begin{pmatrix}
0 & 85 & 60 \\
60 & 0 & 102 \\
115 & 116 & 0
\end{pmatrix}\!,
\begin{pmatrix}
0 & 87 & 48 \\
100 & 0 & 129 \\
100 & 121 & 0
\end{pmatrix}\!,
\begin{pmatrix}
0 & 150 & 67 \\
42 & 0 & 115 \\
107 & 142 & 0
\end{pmatrix}\!
\right\}\!.
\]
Bob sends $M_2$ and $N_2$ to Alice.
\item 
Alice randomly chooses $C_2 \in M_2$ and $D_2 \in N_2$:
\[
C_2 =
\begin{pmatrix}
0 & 70 & 61 \\
119 & 0 & 20 \\
84 & 84 & 0
\end{pmatrix}
\quad
\text{and}
\quad
D_2 =
\begin{pmatrix}
0 & 150 & 67 \\
42 & 0 & 115 \\
107 & 142 & 0
\end{pmatrix}\!.
\]
She computes
\[
U = C_2 \otimes P_1 \otimes W \otimes Q_1 \otimes D_2 =
\begin{pmatrix}
-235 & -247 & -244 \\
-234 & -234 & -231 \\
-254 & -241 & -237
\end{pmatrix}
\]
and sends it to Bob.

\item
Bob randomly chooses $C_1 \in M_1$ and $C_1 \in N_1$:
\[
C_1 =
\begin{pmatrix}
0 & 103 & 134 \\
33 & 0 & 134 \\
27 & 46 & 0
\end{pmatrix}
\quad
\text{and}
\quad
D_1 =
\begin{pmatrix}
0 & 44 & 97 \\
38 & 0 & 33 \\
99 & 144 & 0
\end{pmatrix}\!.
\]
He computes
\[
V = C_1 \otimes P_2 \otimes W \otimes Q_2 \otimes D_1 = 
\begin{pmatrix}
-135 & -168 & -167 \\
-130 & -148 & -147 \\
-180 & -148 & -163
\end{pmatrix}
\]
and sends it to Alice.

\item
Alice computes the secret key
\[
K_A = P_1 \otimes V \otimes Q_1 = 
\begin{pmatrix}
-187 & -258 & -243 \\
-186 & -226 & -275 \\
-182 & -226 & -244
\end{pmatrix}\!.
\]

\item
Bob computes the secret key
\[
K_B = P_2 \otimes U \otimes Q_2 = 
\begin{pmatrix}
-314 & -326 & -323 \\
-294 & -306 & -303 \\
-333 & -331 & -328
\end{pmatrix}\!.
\]
\end{enumerate}
\end{example}

\begin{remark}
For an eavesdropper to break the protocol means being able to compute the key based on the values of $A$, $B$, $W$, $U$, $V$, $M_1$, $M_2$, $N_1$, and $N_2$. 
Note that to find the key, it is not enough to find matrices $P_1'$, $Q_1'$, $C_2'$, and $D_2'$ such that
\begin{align*}
C_2' &\in M_2,\\
D_2' &\in N_2,\\
P_1' \otimes A &= A \otimes P_1', \\
Q_1' \otimes B &= B \otimes Q_1', \\
C_2' \otimes P_1' \otimes W \otimes Q_1' \otimes D_2' &= U
\end{align*}
because $P'_1 \otimes C_1$ must be equal to $P'_1$ and $D_1 \otimes Q_1'$ must be equal to $Q_1'$ for the unknown $C_1$ and $D_1$.
\end{remark}

\begin{protocol}
Alice and Bob agree on a public noncommutative semigroup $\mathcal{S} = \left<S, \cdot, e\right>$, two public pairwise commuting subsets $H \subseteq S$ and $R \subseteq S$, and a public element $w \in S$.

\begin{enumerate}
\item Alice chooses two private elements $p_1 \in H$ and $q_1 \in R$, constructs a public marginal set $M_1 \subseteq S^2$ such that $M_1 \perp p_1\cdot\square \cdot w \cdot \square \cdot q_1$, and sends $M_1$ to Bob.

\item Bob chooses two private elements $p_2 \in H$ and $q_2 \in R$, constructs a public marginal set $M_2 \subseteq S^2$ such that $M_2 \perp p_2 \cdot \square \cdot w \cdot \square \cdot q_2$, and sends $M_2$ to Alice.

\item
Alice chooses a pair $(c_2, d_2) \in M_2$, computes $u = p_1 \cdot c_2 \cdot w \cdot d_2 \cdot q_1$, and sends it to Bob.

\item Bob chooses a pair $(c_1, d_1)$ in $M_1$, computes $v = p_2 \cdot c_1 \cdot w \cdot d_1 \cdot q_2$ and sends it.

\item Alice computes her shared key $k_A = p_1 \cdot v \cdot q_1$.

\item Bob computes his shared key $k_B = p_2 \cdot u \cdot q_2$.
\end{enumerate}
\end{protocol}

They share the same key because
\begin{multline*}
k_A = p_1 \cdot v \cdot q_1 = 
p_1 \cdot (p_2 \cdot c_1 \cdot w \cdot d_1 \cdot q_2) \cdot q_1 = 
(p_1 \cdot p_2) \cdot c_1 \cdot w \cdot d_1 \cdot (q_2 \cdot q_1) = 
(p_2 \cdot p_1) \cdot c_1 \cdot w \cdot d_1 \cdot (q_1 \cdot q_2) =  \\
p_2 \cdot (p_1 \cdot c_1 \cdot w \cdot d_1 \cdot q_1) \cdot q_2 = 
p_2 \cdot (p_1 \cdot w \cdot q_1) \cdot q_2     
\end{multline*}

and
\begin{multline*}
k_B = p_2 \cdot u \cdot q_2  =
p_2 \cdot (p_1 \cdot c_2 \cdot w \cdot d_2 \cdot q_1) \cdot q_2 = 
(p_2 \cdot p_1) \cdot c_2 \cdot w \cdot d_2 \cdot (q_1 \cdot q_2) 
= 
(p_1 \cdot p_2) \cdot c_2 \cdot w \cdot d_2 \cdot (q_2 \cdot q_1) 
=  \\
p_1 \cdot (p_2 \cdot c_2 \cdot w \cdot d_2 \cdot q_2) \cdot q_1 
=
p_1 \cdot (p_2 \cdot w \cdot q_2) \cdot q_1 =
(p_1 \cdot p_2) \cdot w \cdot (q_2 \cdot q_1) =
p_2 \cdot p_1 \cdot w \cdot q_1 \cdot q_2. 
\end{multline*}

\begin{example}
In this example, we also modify the protocol proposed by Grigoriev and Shpilrain in~\cite{GrigorievShpilrain2014}. 

Alice and Bob agree on the semiring $\mathcal{M}^4_{{\min}, {+}}$ and three matrices
\[
A = \begin{pmatrix}
34 & 30 & 9 & 36 \\
91 & 83 & 99 & 57 \\
42 & 72 & 80 & 42 \\
22 & 62 & 67 & 16
\end{pmatrix}\!,
B = \begin{pmatrix}
20 & 32 & 55 & 1 \\
39 & 66 & 67 & 5 \\
95 & 64 & 36 & 95 \\
78 & 91 & 1 & 30
\end{pmatrix}\!,
\text{ and } 
W = \begin{pmatrix}
-2 & -8 & -10 & 4 \\
6 & 4 & 2 & 7 \\
-8 & -2 & -5 & -7 \\
-8 & 10 & 0 & 5
\end{pmatrix}\!.
\]
\begin{enumerate}
\item
Alice randomly selects $p_1(x) =  45 \oplus 98\otimes x \oplus 11 \otimes x^{{\otimes}2}$ and 
$q_1(x) =  77 \oplus 26 \otimes x \oplus 4 \otimes x^{{\otimes 2}}$,
computes 
\[
P_1 = p_1(A) = \begin{pmatrix}
45 & 75 & 54 & 62 \\
90 & 45 & 111 & 84 \\
75 & 83 & 45 & 69 \\
49 & 63 & 42 & 43
\end{pmatrix} \text{ and }
Q_1 = q_1(B) = \begin{pmatrix}
44 & 56 & 6 & 25 \\
63 & 75 & 10 & 31 \\
107 & 90 & 62 & 73 \\
100 & 69 & 27 & 56
\end{pmatrix}\!,
\]
generates the marginal set
\begin{multline*}
M_1 = \left\{\!
\left(\!
\begin{pmatrix}
-25 & 63 & 12 & 37 \\
65 & -25 & 11 & -11 \\
-29 & -39 & -25 & -25 \\
-12 & -40 & 79 & -26
\end{pmatrix}\!, 
\begin{pmatrix}
25 & 21 & 50 & 67 \\
31 & 25 & -1 & 18 \\
33 & 89 & 27 & 20 \\
96 & 20 & 78 & 11
\end{pmatrix}\!
\right)\!\right.,\\
\left.\left(\!
\begin{pmatrix}
-6 & -14 & 96 & 3 \\
16 & -6 & 8 & 69 \\
22 & 17 & -6 & 65 \\
-11 & -3 & -7 & -7
\end{pmatrix}\!,
\begin{pmatrix}
6 & 2 & 90 & -7 \\
12 & 6 & -22 & -1 \\
14 & 8 & -20 & 52 \\
5 & 64 & 64 & 33
\end{pmatrix}\!
\right)\!
\right\}\!,
\end{multline*}
using Algorithm~\ref{AlgoAXBYCeqABC}, and sends this set to Bob.

\item
Bob selects $p_2(x) = 78 \oplus 68 \otimes x \oplus 80 \otimes x^{{\otimes}2}$
and 
$q_2(x) = 45 \oplus 24 \otimes x\oplus 17 \otimes x^{{\otimes}2}$, 
computes
\[
P_2 = p_2(A) = \begin{pmatrix}
78 & 98 & 77 & 104 \\
159 & 78 & 167 & 125 \\
110 & 140 & 78 & 110 \\
90 & 130 & 111 & 78
\end{pmatrix} \text{ and }
Q_2 = q_2(B) = \begin{pmatrix}
44 & 56 & 19 & 25 \\
63 & 45 & 23 & 29 \\
119 & 88 & 45 & 86 \\
102 & 82 & 25 & 45
\end{pmatrix}\!,
\]
generates the marginal set 
\begin{multline*}
M_2 = 
\left\{\!\left(\!
\begin{pmatrix}
59 & 69 & 58 & 58 \\
71 & 59 & 73 & 73 \\
65 & 84 & 59 & 59 \\
97 & 58 & 92 & 59
\end{pmatrix}\!, 
\begin{pmatrix}
-59 & 39 & -59 & 1 \\
-57 & -59 & 11 & -58 \\
9 & 5 & 5 & -61 \\
92 & -39 & -81 & 98
\end{pmatrix}\!\right)\right.\!,  \\
\left.
\left(\!
\begin{pmatrix}
-70 & -81 & 49 & -7 \\
3 & -70 & 45 & -42 \\
43 & -42 & -70 & -62 \\
-59 & 47 & 97 & -70
\end{pmatrix}\!, 
\begin{pmatrix}
70 & 81 & 44 & 64 \\
72 & 70 & 46 & 66 \\
74 & 75 & 48 & 68 \\
69 & 75 & 85 & 63
\end{pmatrix}\!
\right)\!\right\},
\end{multline*}
using Algorithm~\ref{AlgoAXBYCeqABC}, and sends this set to Alice.

\item
Alice randomly selects the tuple
\[
(C_2, D_2) = \left(\!\begin{pmatrix}
-70 & -81 & 49 & -7 \\
3 & -70 & 45 & -42 \\
43 & -42 & -70 & -62 \\
-59 & 47 & 97 & -70
\end{pmatrix}\!,\begin{pmatrix}
70 & 81 & 44 & 64 \\
72 & 70 & 46 & 66 \\
74 & 75 & 48 & 68 \\
69 & 75 & 85 & 63
\end{pmatrix}\!\right)\!,
\]
computes \[
U = P_1 \otimes C_2 \otimes W \otimes D_2 \otimes Q_1 = \begin{pmatrix}
83 & 95 & 45 & 64 \\
95 & 107 & 57 & 76 \\
81 & 93 & 43 & 62 \\
78 & 90 & 40 & 59
\end{pmatrix}\!,
\]
and sends it to Bob.

\item 
Bob selects the tuple
\[
(C_1, D_1) = \left(\!\begin{pmatrix}
-25 & 63 & 12 & 37 \\
65 & -25 & 11 & -11 \\
-29 & -39 & -25 & -25 \\
-12 & -40 & 79 & -26
\end{pmatrix}\!,\begin{pmatrix}
25 & 21 & 50 & 67 \\
31 & 25 & -1 & 18 \\
33 & 89 & 27 & 20 \\
96 & 20 & 78 & 11
\end{pmatrix}\!\right)\!,
\]
computes \[
V = P_2 \otimes C_1 \otimes W \otimes D_1 \otimes Q_2 = \begin{pmatrix}
113 & 110 & 81 & 94 \\
128 & 125 & 96 & 109 \\
114 & 111 & 82 & 95 \\
113 & 110 & 81 & 94
\end{pmatrix}\!,
\]
and sends it to Alice.

\item
Alice computes her shared key:
\[
K_A = P_1 \otimes V \otimes Q_1 = \begin{pmatrix}
202 & 208 & 164 & 183 \\
217 & 223 & 179 & 198 \\
203 & 209 & 165 & 184 \\
200 & 206 & 162 & 181
\end{pmatrix}\!.
\]
\item 
Bob computes his shared key:
\[
K_B = P_2 \otimes U \otimes Q_2 = \begin{pmatrix}
202 & 208 & 164 & 183 \\
217 & 223 & 179 & 198 \\
203 & 209 & 165 & 184 \\
200 & 206 & 162 & 181
\end{pmatrix}\!.
\]
\end{enumerate}
\end{example}

\begin{remark}
For an eavesdropper to break the protocol means to be able to compute the key based on the values of $A$, $B$, $W$, $U$, $V$, $M_1$ and $M_2$. 
Note that to find the key, it is not enough to find matrices $P_1'$, $Q_1'$, $C_2'$, and $D_2'$ such that
\begin{align*}
(C_2', D_2') &\in M_2,\\
P_1' \otimes A &= A \otimes P_1', \\
Q_1' \otimes B &= B \otimes Q_1', \\
P_1' \otimes C_2' \otimes W \otimes D_2' \otimes Q_1' &= U
\end{align*}
because $P'_1 \otimes C_1 \otimes W \otimes D_2 \otimes Q_1'$ must be equal to $P'_1 \otimes W  \otimes Q_1'$ for the unknown $C_1$ and $D_2$.
\end{remark}

\begin{protocol}
Alice and Bob agree on a public noncommutative semigroup $\mathcal{S} = \left<S, \cdot, e\right>$, four public pairwise commuting subssets
$H_1 \subseteq S$, $R_1 \subseteq S$, $H_2 \subseteq S$, and $R_2 \subseteq S$.
and two public elements $w_1, w_2 \in S$.
\begin{enumerate}
\item Alice selects private elements $p_{11} \in H_1$, $p_{12} \in H_2$, $q_{11} \in R_1$, and $q_{12} \in R_2$, constructs public marginal sets $M_{11} \subseteq S$, $M_{12} \subseteq S^2$, and $M_{13} \subseteq S$ such that
$M_{11} \perp p_{11} \cdot \square$, $M_{12} \perp \square_1 \cdot  q_{11} \cdot p_{12} \cdot \square_2$, and $M_{13} \perp \square \cdot q_{12}$, and sends $M_{11}$, $M_{12}$, and $M_{13}$ to Bob.

\item Bob selects private elements $p_{21} \in H_1$, $p_{22} \in H_2$, $q_{21} \in R_1$, and $q_{22} \in R_2$, constructs public marginal sets $M_{21} \subseteq S$, $M_{22} \subseteq S^2$, and $M_{23} \subseteq S$ such that
$M_{21} \perp p_{21} \cdot \square$, $M_{22} \perp \square_1 \cdot  q_{21} \cdot p_{22} \cdot \square_2$, and $M_{23} \perp \square \cdot q_{22}$, and sends $M_{21}$, $M_{22}$, and $M_{23}$ to Alice.

\item Alice selects private elements $c_{21} \in M_{21}$, $(d_{21}, c_{22}) \in M_{22}$, and $d_{22} \in M_{23}$, computes 
$$u = (c_{21} \cdot p_{11} \cdot w_1 \cdot q_{11} \cdot d_{21}, c_{22} \cdot p_{12} \cdot w_2 \cdot q_{12} \cdot d_{22}),$$ and sends $u$ to Bob.

\item Bob selects private elements $c_{11} \in M_{11}$, $(d_{11}, c_{12}) \in M_{12}$, and $d_{12} \in M_{13}$, computes 
$$v = (c_{11} \cdot p_{21} \cdot w_1 \cdot q_{21} \cdot d_{11}, c_{12} \cdot p_{22} \cdot w_2 \cdot q_{22} \cdot d_{12}),$$ and sends $v$ to Alice.

\item Alice computes her shared secret key
$k_A = p_{11} \cdot v_1 \cdot q_{11} \cdot p_{12} \cdot v_2 \cdot q_{12}
$.

\item Bob computes his shared secret key
$k_B = p_{21} \cdot u_1 \cdot q_{21} \cdot p_{22} \cdot u_2 \cdot q_{22}
$.
\end{enumerate}
\end{protocol}

They share the same key because 
\begin{multline*}
k_A = 
p_{11} \cdot v_1 \cdot q_{11} \cdot p_{12} \cdot v_2 \cdot q_{12} = 
p_{11} \cdot (c_{11} \cdot p_{21} \cdot w_1 \cdot q_{21} \cdot d_{11}) \cdot q_{11} \cdot p_{12} \cdot (c_{12} \cdot p_{22} \cdot w_2 \cdot q_{22} \cdot d_{12}) \cdot q_{12} =  \\
(p_{11} \cdot c_{11}) \cdot p_{21} \cdot w_1 \cdot q_{21} \cdot (d_{11} \cdot q_{11} \cdot p_{12} \cdot c_{12}) \cdot p_{22} \cdot w_2 \cdot q_{22} \cdot (d_{12} \cdot q_{12}) =
p_{11} \cdot p_{21} \cdot w_1 \cdot q_{21} \cdot q_{11} \cdot p_{12} \cdot p_{22} \cdot w_2 \cdot q_{22} \cdot q_{12}
\end{multline*}
and
\begin{multline*}
k_B = p_{21} \cdot u_1 \cdot q_{21} \cdot p_{22} \cdot u_2 \cdot q_{22} 
= 
p_{21} \cdot (c_{21} \cdot p_{11} \cdot w_1 \cdot q_{11} \cdot d_{21}) \cdot q_{21} \cdot p_{22} \cdot (c_{22} \cdot p_{12} \cdot w_2 \cdot q_{12} \cdot d_{22}) \cdot q_{22} = \\
(p_{21} \cdot c_{21}) \cdot p_{11} \cdot w_1 \cdot q_{11} \cdot (d_{21} \cdot q_{21} \cdot p_{22} \cdot c_{22}) \cdot p_{12} \cdot w_2 \cdot q_{12} \cdot (d_{22} \cdot q_{22}) = 
p_{21} \cdot p_{11} \cdot w_1 \cdot q_{11} \cdot q_{21} \cdot p_{22} \cdot p_{12} \cdot w_2 \cdot q_{12} \cdot q_{22} = \\
(p_{21} \cdot p_{11}) \cdot w_1 \cdot (q_{11} \cdot q_{21}) \cdot (p_{22} \cdot p_{12}) \cdot w_2 \cdot (q_{12} \cdot q_{22}) =
(p_{11} \cdot p_{21}) \cdot w_1 \cdot (q_{21} \cdot q_{11}) \cdot (p_{12} \cdot p_{22}) \cdot w_2 \cdot (q_{22} \cdot q_{12}).
\end{multline*}

\begin{example}
Alice and Bob agree on the semigroup of tropical matrices of size $3 \times 3$ over the min-plus algebra $\mathcal{M}_{\min, {+}}^3$, four sets  of Linde--de la Puente matrices $H_1$, $H_2$, $R_1$, and $R_2$, and
\[
W_1 =
\begin{pmatrix}
80 & 7 & 64 \\
46 & 57 & 15 \\
21 & 36 & 7
\end{pmatrix}\! \quad \text{and} \quad
W_2 =
\begin{pmatrix}
5 & 3 & 68 \\
95 & 89 & 34 \\
99 & 21 & 86
\end{pmatrix}\!.
\]
\begin{enumerate}
\item Alice chooses four private matrices:
\[
P_{11} =
\begin{pmatrix}
-15 & 126 & 166 \\
124 & -15 & 164 \\
153 & 142 & -15
\end{pmatrix}\!,
Q_{11} =
\begin{pmatrix}
-39 & 99 & 153 \\
97 & -39 & 96 \\
101 & 136 & -39
\end{pmatrix}\!,
\]
\[
P_{12} =
\begin{pmatrix}
-3 & 61 & 33 \\
33 & -3 & 36 \\
51 & 45 & -3
\end{pmatrix}\!,
Q_{12} =
\begin{pmatrix}
-64 & 93 & 123 \\
95 & -64 & 68 \\
66 & 101 & -64
\end{pmatrix}\!.
\]
She constructs three public marginal sets using Algorithms~\ref{AlgoAXeqA2}, \ref{AlgoXAeqA2}, and \ref{AlgoAXBYCeqABC}
\[
M_{11} =
\left\{\!
\begin{pmatrix}
0 & 169 & 181 \\
184 & 0 & 200 \\
188 & 194 & 0
\end{pmatrix}\!,
\begin{pmatrix}
0 & 195 & 187 \\
142 & 0 & 184 \\
192 & 191 & 0
\end{pmatrix}\!
\right\}\!,
\]

\[
M_{12} =
\left\{\!
\left(\!
\begin{pmatrix}
20 & 84 & 56 \\
56 & 20 & 59 \\
74 & 77 & 20
\end{pmatrix}\!,
\begin{pmatrix}
-20 & 47 & 33 \\
16 & -20 & 19 \\
49 & 28 & -20
\end{pmatrix}\!
\right)\!,
\left(\!
\begin{pmatrix}
7 & 76 & 43 \\
43 & 7 & 46 \\
74 & 55 & 7
\end{pmatrix}\!,
\begin{pmatrix}
-7 & 57 & 29 \\
56 & -7 & 32 \\
47 & 41 & -7
\end{pmatrix}\!
\right)\!
\right\}\!,
\]
and
\[
M_{13} =
\left\{\!
\begin{pmatrix}
0 & 183 & 189 \\
198 & 0 & 180 \\
170 & 173 & 0
\end{pmatrix}\!,
\begin{pmatrix}
0 & 169 & 193 \\
178 & 0 & 168 \\
199 & 184 & 0
\end{pmatrix}\!
\right\}\!.
\]
She sends $M_{11}$, $M_{12}$, and $M_{13}$ to Bob.
\item
Bob chooses four private matrices:
\[
P_{21} =
\begin{pmatrix}
-77 & 12 & 14 \\
14 & -77 & 19 \\
16 & 13 & -77
\end{pmatrix}\!,
Q_{21} =
\begin{pmatrix}
-82 & 19 & 20 \\
18 & -82 & 19 \\
12 & 16 & -82
\end{pmatrix}\!, 
\]
\[
P_{22} =
\begin{pmatrix}
-68 & 42 & 43 \\
39 & -68 & 45 \\
47 & 37 & -68
\end{pmatrix}\!,
Q_{22} =
\begin{pmatrix}
-8 & 26 & 38 \\
37 & -8 & 34 \\
32 & 27 & -8
\end{pmatrix}\!.
\]
He constructs three public marginal sets using Algorithms~\ref{AlgoAXeqA2}, \ref{AlgoXAeqA2}, and \ref{AlgoAXBYCeqABC}:
\[
M_{21} =
\left\{\!
\begin{pmatrix}
0 & 109 & 118 \\
113 & 0 & 195 \\
138 & 186 & 0
\end{pmatrix}\!,
\begin{pmatrix}
0 & 174 & 103 \\
178 & 0 & 103 \\
106 & 91 & 0
\end{pmatrix}\!
\right\}\!,
\]
\[
M_{22} =
\left\{\!
\left(\!
\begin{pmatrix}
-17 & 84 & 85 \\
83 & -17 & 84 \\
77 & 81 & -17
\end{pmatrix}\!,
\begin{pmatrix}
17 & 118 & 119 \\
117 & 17 & 118 \\
111 & 115 & 17
\end{pmatrix}\!
\right)\!,
\left(\!
\begin{pmatrix}
39 & 140 & 141 \\
139 & 39 & 140 \\
133 & 137 & 39
\end{pmatrix}\!,
\begin{pmatrix}
-39 & 62 & 63 \\
61 & -39 & 62 \\
55 & 59 & -39
\end{pmatrix}\!
\right)\!
\right\}\!,
\]
and
\[
M_{23} =
\left\{\!
\begin{pmatrix}
0 & 194 & 167 \\
171 & 0 & 49 \\
78 & 139 & 0
\end{pmatrix}\!,
\begin{pmatrix}
0 & 154 & 175 \\
102 & 0 & 102 \\
133 & 43 & 0
\end{pmatrix}\!
\right\}.
\]
He sends these sets to Alice.
\item
Alice selects 
\[
C_{21} =
\begin{pmatrix}
0 & 174 & 103 \\
178 & 0 & 103 \\
106 & 91 & 0
\end{pmatrix}\!,
D_{21}= \begin{pmatrix}
39 & 140 & 141 \\
139 & 39 & 140 \\
133 & 137 & 39
\end{pmatrix}\!,
\]
\[
C_{22} = 
\begin{pmatrix}
-39 & 62 & 63 \\
61 & -39 & 62 \\
55 & 59 & -39
\end{pmatrix}\!,
D_{22} =
\begin{pmatrix}
0 & 194 & 167 \\
171 & 0 & 49 \\
78 & 139 & 0
\end{pmatrix}\!,
\]
computes
\[
U =
(C_{21} \otimes P_{11} \otimes W_1 \otimes Q_{11} \otimes D_{21}, C_{22} \otimes P_{12} \otimes W_2 \otimes Q_{12} \otimes D_{22}) = 
\left(\!
\begin{pmatrix}
65 & -8 & 49 \\
31 & 42 & 0 \\
6 & 21 & -8
\end{pmatrix}\!,
\begin{pmatrix}
-101 & -103 & -54 \\
-65 & -67 & -72 \\
-47 & -85 & -36
\end{pmatrix}\!\right),
\]
and sends it to Bob.

\item
Bob selects
\[
C_{11} =
\begin{pmatrix}
0 & 195 & 187 \\
142 & 0 & 184 \\
192 & 191 & 0
\end{pmatrix}\!,
D_{11} =
\!\begin{pmatrix}
7 & 76 & 43 \\
43 & 7 & 46 \\
74 & 55 & 7
\end{pmatrix}\!,
C_{12} =
\begin{pmatrix}
-7 & 57 & 29 \\
56 & -7 & 32 \\
47 & 41 & -7
\end{pmatrix}\!,
D_{12} =
\begin{pmatrix}
0 & 183 & 189 \\
198 & 0 & 180 \\
170 & 173 & 0
\end{pmatrix}\!,
\]
computes
\begin{multline*}
V =
(C_{11} \otimes P_{21} \otimes W_1 \otimes Q_{21} \otimes D_{11}, C_{12} \otimes P_{22} \otimes W_2 \otimes Q_{22} \otimes D_{12}) = \\
\left(\!
\begin{pmatrix}
-109 & -145 & -106 \\
-106 & -95 & -137 \\
-131 & -116 & -145
\end{pmatrix}\!, 
\begin{pmatrix}
-78 & -80 & -38 \\
-15 & -23 & -49 \\
-24 & -62 & -20
\end{pmatrix}\!
\right)\!,
\end{multline*}
and sends it to Alice.
\item 
Alice computes the shared key:
\[
K_A = P_{11} \otimes V_1 \otimes Q_{11} \otimes P_{12} \otimes V_2 \otimes Q_{12} = 
\begin{pmatrix}
-308 & -310 & -315 \\
-305 & -320 & -278 \\
-330 & -332 & -290
\end{pmatrix}\!.
\]
\item 
Bob computes the shared key:
\[
K_B = P_{21} \otimes U_1 \otimes Q_{21} \otimes P_{22} \otimes U_2 \otimes Q_{22} = 
\begin{pmatrix}
-308 & -310 & -315 \\
-305 & -320 & -278 \\
-330 & -332 & -290
\end{pmatrix}.
\]
\end{enumerate}
\end{example}

\begin{remark}
For an eavesdropper to break the protocol means being able to compute the key based on the values of $W_1$, $W_2$, $U$, $V$, $M_{11}$, $M_{12}$, $M_{13}$, $M_{21}$, $M_{22}$, and $M_{23}$. Note that to find the key, it is not enough to find matrices
$P_{11}'$, $Q_{11}'$, $P_{12}'$, $Q_{12}'$, $C_{21}'$, $D_{21}'$, $C_{22}'$, $D_{22}'$ such that
\begin{align*}
C_{21}' &\in M_{21}, \\
(D_{21}', C_{22}') &\in M_{22}, \\
D_{22}' &\in M_{23}, \\
C'_{21} \otimes P'_{11} \otimes W_1 \otimes Q'_{11} \otimes D'_{21} &= U_1, \\
C'_{22} \otimes P'_{12} \otimes W_2 \otimes Q'_{12} \otimes D'_{22} &= U_2
\end{align*}
because $P_{11}' \otimes C_{11}$ must be equal to $P_{11}'$, $D_{11} \otimes Q_{11}' \otimes P_{12}' \otimes C_{12}$ must be equal to $Q_{11}' \otimes P_{12}'$, and $D_{12} \otimes Q_{12}'$ must be equal to $Q_{12}'$ for the unknown $C_{11}$, $D_{11}$, $C_{12}$, and $D_{12}$. 
\end{remark}

This protocol can be generalized.

\begin{protocol}
Alice and Bob agree on a public noncommutative semigroup $\mathcal{S} = \left<S, \cdot, e\right>$, public pairwise commuting subssets
$H_i \subseteq S$, $i \in \{1, \ldots, n\}$, $R_i \subseteq S$, $i \in \{1, \ldots n\}$, and public elements $w_i \in S$, $i \in \{1, \ldots, n\}$.
\begin{enumerate}
\item Alice selects private elements $p_{1i} \in H_i$, $i \in \{1, \ldots, n\}$, and $q_{1i} \in R_i$, $i \in \{1, \ldots n\}$, constructs public marginal sets $M_{11} \subseteq S$, $M_{1i} \subseteq S^2$, $i \in \{2, \ldots, n\}$, and $M_{1,n+1} \subseteq S$ such that
$M_{11} \perp p_{11} \cdot \square$, $M_{1i} \perp \square_1 \cdot  q_{1,i-1} \cdot p_{1i} \cdot \square_2$, $i \in \{2, \ldots n\}$, and $M_{1,n+1} \perp \square \cdot q_{1n}$, and sends $M_{1i}$, $i \in \{1, \ldots, n  +1\}$,  to Bob.

\item Bob selects private elements $p_{2i} \in H_i$, $i \in \{1, \ldots, n\}$, and $q_{2i} \in R_i$, $i \in \{1, \ldots n\}$, constructs public marginal sets $M_{21} \subseteq S$, $M_{2i} \subseteq S^2$, $i \in \{2, \ldots, n\}$, and $M_{2,n+1} \subseteq S$ such that
$M_{21} \perp p_{21} \cdot \square$, $M_{2i} \perp \square_1 \cdot  q_{2,i-1} \cdot p_{2i} \cdot \square_2$, $i \in \{2, \ldots n\}$, and $M_{2,n+1} \perp \square \cdot q_{2n}$, and sends $M_{2i}$, $i \in \{1, \ldots, n  +1\}$,  to Alice.

\item Alice selects private elements $c_{21} \in M_{21}$, $(d_{2,i-1}, c_{2i}) \in M_{2i}$, $i \in \{2, \ldots n\}$, and $d_{2n} \in M_{2,n+1}$, computes 
$$u = (c_{21} \cdot p_{11} \cdot w_1 \cdot q_{11} \cdot d_{21}, \ldots, c_{2n} \cdot p_{1n} \cdot w_n \cdot q_{1n} \cdot d_{2n}), $$ and sends $u$ to Bob.

\item Bob selects private elements $c_{11} \in M_{11}$, $(d_{1,i-1}, c_{1i}) \in M_{1i}$, $i \in \{2, \ldots n\}$, and $d_{1n} \in M_{1,n+1}$, computes 
$$v = (c_{11} \cdot p_{21} \cdot w_1 \cdot q_{21} \cdot d_{11}, \ldots
c_{1n} \cdot p_{2n} \cdot w_n \cdot q_{2n} \cdot d_{1n}),$$ and sends $v$ to Alice.

\item Alice computes her shared secret key
$k_A = \prod_{i=1}^n p_{1i} \cdot v_i \cdot q_{1i}$.

\item Bob computes his shared secret key
$k_B = \prod_{i=1}^n p_{2i} \cdot u_i \cdot q_{2i}$.
\end{enumerate}
\end{protocol}

\section{Conclusion}
In this paper, we introduced the definitions of marginal sets for semigroups and semirings. 
We considered several examples of marginal sets for tropical matrix algebras. 
Also, we demonstrated how this concept can be used to improve some key exchange schemes.

For future work, it would be interesting to explore examples of marginal sets in other classes of semigroups.

Another question to investigate is to find out what information a $w(A_1, \ldots, A_m,\allowbreak  \square_1, \ldots, \square_n)$-marginal set reveals about the matrices $A_1, \ldots, A_m$.

\section*{Funding}
The research of the first author was supported by the State Contract of the Sobolev Institute of Mathematics (Project FWNF-2026-0033).

\renewcommand*{\bibfont}{\small}
\printbibliography

\end{document}